\PassOptionsToPackage{unicode}{hyperref}
\PassOptionsToPackage{hyphens}{url}
\PassOptionsToPackage{dvipsnames,svgnames,x11names}{xcolor}
\documentclass[
  12pt]{article}

\usepackage{amssymb}
\usepackage{mathrsfs}
\usepackage{iftex}

\usepackage{dsfont} 
\usepackage{graphicx,psfrag,epsf}
\usepackage{enumerate}
\usepackage{array}
\usepackage{url} 
\usepackage{amsmath}
\usepackage{tcolorbox}
\usepackage{subcaption}
\usepackage{bm}
\usepackage{arydshln}
\usepackage{algorithm,algorithmic}
\usepackage{enumitem}
\usepackage{caption}
\usepackage{xr-hyper}
\usepackage{hyperref}

\newtheorem{assumption}{Assumption}
\newcommand{\piSyst}{\pi_{\cM}(\mathscr{A}_{\cQ})}

\newcommand{\cMQ}{\mathcal{M}\text{-}\mathcal{Q}}

\let\hat\widehat
\let\tilde\widetilde

\newcommand{\ba}{\bm{a}}

\newcommand{\bbf}{\bm{f}}

\newcommand{\bx}{\bm{x}}

\newcommand{\Ab}{\mathbf{A}}

\newcommand{\Fb}{\mathbf{F}}
\newcommand{\Gb}{\mathbf{G}}
\newcommand{\Hb}{\mathbf{H}}
\newcommand{\Ib}{\mathbf{I}}

\newcommand{\Lb}{\mathbf{L}}
\newcommand{\Mb}{\mathbf{M}}

\newcommand{\Sbb}{\mathbf{S}}

\newcommand{\Yb}{\mathbf{Y}}

\newcommand{\bF}{\bm{F}}

\newcommand{\cA}{\mathcal{A}}

\newcommand{\cD}{\mathcal{D}}

\newcommand{\cJ}{\mathcal{J}}

\newcommand{\cL}{\mathcal{L}}
\newcommand{\cM}{\mathcal{M}}
\newcommand{\cN}{\mathcal{N}}

\newcommand{\cQ}{\mathcal{Q}}

\newcommand{\cS}{{\mathcal{S}}}
\newcommand{\cT}{{\mathcal{T}}}

\newcommand{\EE}{\mathbb{E}}

\newcommand{\RR}{\mathbb{R}}

\newcommand{\bbeta}{\bm{\beta}}

\newcommand{\blambda}{\bm{\lambda}}

\newcommand{\bxi}{\bm{\xi}}

\newcommand{\bpsi}{\bm{\psi}}

\newcommand{\bTheta}{\bm{\Theta}}
\newcommand{\bLambda}{\bm{\Lambda}}

\newcommand{\bSigma}{\bm{\Sigma}}

\newcommand{\bPhi}{\bm{\Phi}}
\newcommand{\bPsi}{\bm{\Psi}}

\newcommand*{\zero}{{\bm 0}}
\newcommand*{\one}{{\bm 1}}

\def\T{{ \intercal }}

\ifx\BlackBox\undefined
\newcommand{\BlackBox}{\rule{1.5ex}{1.5ex}}  
\fi

\ifx\QED\undefined
\def\QED{~\rule[-1pt]{5pt}{5pt}\par\medskip}
\fi

\ifx\proof\undefined
\newenvironment{proof}{\par\noindent{\bf Proof\ }}{\hfill\BlackBox\\[2mm]}
\fi

\ifx\theorem\undefined
\newtheorem{theorem}{Theorem}
\fi
\ifx\example\undefined
\newtheorem{example}{Example}
\fi
\ifx\property\undefined
\newtheorem{property}{Property}
\fi
\ifx\lemma\undefined
\newtheorem{lemma}{Lemma}
\fi
\ifx\proposition\undefined
\newtheorem{proposition}{Proposition}
\fi
\ifx\remark\undefined
\newtheorem{remark}{Remark}
\fi
\ifx\corollary\undefined
\newtheorem{corollary}{Corollary}
\fi
\ifx\definition\undefined
\newtheorem{definition}{Definition}
\fi
\ifx\conjecture\undefined
\newtheorem{conjecture}{Conjecture}
\fi
\ifx\fact\undefined
\newtheorem{fact}{Fact}
\fi
\ifx\claim\undefined
\newtheorem{claim}{Claim}
\fi
\ifx\cond\undefined
\newtheorem{cond}{Condition}
\fi

\ifPDFTeX
  \usepackage[T1]{fontenc}
  \usepackage[utf8]{inputenc}
  \usepackage{textcomp} 
\else 
  \usepackage{unicode-math}
  \defaultfontfeatures{Scale=MatchLowercase}
  \defaultfontfeatures[\rmfamily]{Ligatures=TeX,Scale=1}
\fi
\usepackage{lmodern}
\ifPDFTeX\else  
\fi
\IfFileExists{upquote.sty}{\usepackage{upquote}}{}
\IfFileExists{microtype.sty}{
  \usepackage[]{microtype}
  \UseMicrotypeSet[protrusion]{basicmath} 
}{}
\makeatletter
\@ifundefined{KOMAClassName}{
  \IfFileExists{parskip.sty}{%
    \usepackage{parskip}
  }{
    \setlength{\parindent}{2pt}
    \setlength{\parskip}{6pt plus 2pt minus 1pt}}
}{
  \KOMAoptions{parskip=half}}
\makeatother
\usepackage{xcolor}
\makeatletter
\ifx\paragraph\undefined\else
  \let\oldparagraph\paragraph
  \renewcommand{\paragraph}{
    \@ifstar
      \xxxParagraphStar
      \xxxParagraphNoStar
  }
  \newcommand{\xxxParagraphStar}[1]{\oldparagraph*{#1}\mbox{}}
  \newcommand{\xxxParagraphNoStar}[1]{\oldparagraph{#1}\mbox{}}
\fi
\ifx\subparagraph\undefined\else
  \let\oldsubparagraph\subparagraph
  \renewcommand{\subparagraph}{
    \@ifstar
      \xxxSubParagraphStar
      \xxxSubParagraphNoStar
  }
  \newcommand{\xxxSubParagraphStar}[1]{\oldsubparagraph*{#1}\mbox{}}
  \newcommand{\xxxSubParagraphNoStar}[1]{\oldsubparagraph{#1}\mbox{}}
\fi
\makeatother

\usepackage{longtable,booktabs,array}
\usepackage{calc} 
\usepackage{etoolbox}
\makeatletter
\patchcmd\longtable{\par}{\if@noskipsec\mbox{}\fi\par}{}{}
\makeatother
\IfFileExists{footnotehyper.sty}{\usepackage{footnotehyper}}{\usepackage{footnote}}
\makesavenoteenv{longtable}
\usepackage{graphicx}
\makeatletter
\def\maxwidth{\ifdim\Gin@nat@width>\linewidth\linewidth\else\Gin@nat@width\fi}
\def\maxheight{\ifdim\Gin@nat@height>\textheight\textheight\else\Gin@nat@height\fi}
\makeatother
\setkeys{Gin}{width=\maxwidth,height=\maxheight,keepaspectratio}
\makeatletter
\def\fps@figure{htbp}
\makeatother

\makeatletter
\@ifpackageloaded{caption}{}{\usepackage{caption}}
\AtBeginDocument{%
\ifdefined\contentsname
  \renewcommand*\contentsname{Table of contents}
\else
  \newcommand\contentsname{Table of contents}
\fi
\ifdefined\listfigurename
  \renewcommand*\listfigurename{List of Figures}
\else
  \newcommand\listfigurename{List of Figures}
\fi
\ifdefined\listtablename
  \renewcommand*\listtablename{List of Tables}
\else
  \newcommand\listtablename{List of Tables}
\fi
\ifdefined\figurename
  \renewcommand*\figurename{Figure}
\else
  \newcommand\figurename{Figure}
\fi
\ifdefined\tablename
  \renewcommand*\tablename{Table}
\else
  \newcommand\tablename{Table}
\fi
}
\@ifpackageloaded{float}{}{\usepackage{float}}
\floatstyle{ruled}
\@ifundefined{c@chapter}{\newfloat{codelisting}{h}{lop}}{\newfloat{codelisting}{h}{lop}[chapter]}
\floatname{codelisting}{Listing}

\makeatother
\makeatletter
\@ifpackageloaded{caption}{}{\usepackage{caption}}
\@ifpackageloaded{subcaption}{}{\usepackage{subcaption}}
\makeatother

\ifLuaTeX
  \usepackage{selnolig}  
\fi
\usepackage[]{natbib}
\usepackage{bookmark}

\IfFileExists{xurl.sty}{\usepackage{xurl}}{} 
\hypersetup{
  pdftitle={Title},
  pdfauthor={Author 1; Author 2},
  pdfkeywords={3 to 6 keywords, that do not appear in the title},
  colorlinks=true,
  linkcolor={blue},
  filecolor={Maroon},
  citecolor={Blue},
  urlcolor={Blue},
  pdfcreator={LaTeX via pandoc}}

\begin{document}

\def\spacingset#1{\renewcommand{\baselinestretch}%
{#1}\small\normalsize} \spacingset{1}


  \title{\bf  Statistical Analysis of Block Structured Latent Variable Models}
  \author{Chengyu Cui and Gongjun Xu\hspace{.2cm}\\
    Department of Statistics, University of Michigan}
    \date{}
  \maketitle

\bigskip
\begin{abstract}
This paper studies block structured latent variable models, in which observed variables are grouped into distinct blocks based on their relationships with the underlying latent variables. These block structures are prevalent in various fields such as psychology, education, economics, and genetics. Despite their widespread applications, the fundamental statistical properties of these models remain largely unexplored.
In this work, we present a comprehensive statistical analysis of the block structured latent variable models. In particular, we first derive conditions for model identifiability across various block designs. 
Furthermore, we investigate the maximum likelihood estimation under these identifiability constraints. To accommodate these intricate constraints associated with various block configurations, we introduce a Lagrangian-type formulation for the constrained nonconvex optimization problem and show that its optimum coincides with that of the original problem. This formulation serves as a critical tool for understanding the behavior of the constrained estimator under various block structures. 
Building on that, we establish sharp non-asymptotic error bounds and asymptotic distributions of the constrained maximum likelihood estimator. We also propose a computational framework to obtain the estimator and establish theoretical properties for the algorithm output. Our theoretical findings are validated through simulation studies and empirical data analyses. 
\end{abstract}

\noindent%

\spacingset{1.19}

\section{Introduction}\label{sec_intro}

Latent variable models are widely used statistical and machine learning tools in many scientific fields for analyzing high-dimensional data, with dependencies among observed variables captured by latent variables, also called latent factors~\citep{andersonintroduction,bartholomew2011latent}. In many scientific applications, the data are collected from multiple domains or sources, hereafter referred to as blocks, including psychology~\citep{thompson2004exploratory,chen2020structured,fang2021identifiability}, education~\citep{brunner2008no,reckase2009,gu2020partial}, economics~\citep{goyal2008common,ando2017clustering} and genetics~\citep{wiwie2015comparing,bing2020adaptive,zhang2020imputed,xue2021integrating,zhou2023multi}.
To capture within-block homogeneity and between-block heterogeneity, it is common to impose a block structured loading pattern in latent variable models~\citep{loehlin2004latent,reckase2009}. Under this specification, response variables within the same block are influenced by a specific subset of latent variables. 
Such block structured latent variable models have found extensive applications, such as the following examples.

     {\bf Psychological/Educational Measurement.} Block structured latent variable models have been a popular tool in measuring psychological and educational latent traits, such as personality and ability~\citep{thompson2004exploratory,reckase2009,cai2010two,chen2020structured,gu2023joint}. In these applications, participants typically complete questionnaires organized into distinct blocks, with each block specifically designed to measure one or a subset of the latent variables. Such structured design enhances the scientific interpretability of the latent variables, making the resulting inferences more closely aligned with the underlying psychological or educational constructs.

    {\bf Economics.} Block structured latent variable models have also been widely used in economics where economic indicators, such as asset returns, GDP, and credit risk, are often categorized into distinct groups~\citep{goyal2008common}. In these settings, some latent variables are common and affect all observations, while others are specific to certain groups, thereby accounting for the heterogeneity across groups. Block structured latent variable models effectively capture these structures and have extensive applications in areas such as asset pricing~\citep{ando2017clustering} and macroeconomic analysis~\citep{andreou2019inference}.

    {\bf Data Integration.} Block structured latent variable models have also gained popularity in data integration~\citep{lock2013joint,feng2018angle,zhu2020generalized,xue2021integrating}, or referred to as 
    multi-view data integration~\citep{li2018review}. When integrating multiple data sources, common latent factors are often assumed to account for the shared variation across all sources, while source-specific factors capture individual variations. These latent factors together explain both joint and unique patterns in the data, thereby modeling the homogeneity and heterogeneity of the sources simultaneously. This approach has found widespread applications in different fields such as genetics ~\citep{lock2013joint,li2018review}, omics~\citep{cantini2021benchmarking}, and public health~\citep{xue2021integrating}.  

\medskip

Despite the widespread application of block structured latent variable models, many of their statistical properties remain largely unexplored, including model identifiability, consistency of the maximum likelihood estimation under the identifiability constraints, and related statistical inference for the model parameters and latent factors. \begin{enumerate}[label=$(\roman*)$,leftmargin=0.7cm]

\item\label{prblm_id}
    For the identifiability of the model parameters and latent factors, most of the recent studies have focused on certain special cases of the block structured models~\citep{ando2017clustering,chen2020structured,bing2020adaptive,fang2021identifiability,qiao2025exploratory}. For instance, \cite{chen2020structured} investigated the identifiability of the latent factors in an asymptotic setting, but their assumptions are often too restrictive for the block structures commonly encountered in applications such as data integration~\citep{lock2013joint,feng2018angle}. \citet{bing2020adaptive} examined identifiability for structured linear latent variable models under the presence of ``pure variables'', which may not be satisfied in many psychological and educational applications~\citep{cai2010two,fang2021identifiability}.
    In addition, \cite{fang2021identifiability} and \cite{qiao2025exploratory} established identifiability conditions for bi-factor models and hierarchical latent factor models, two specific types of block structured latent variable models. Despite these recent developments, identifiability results for the general family of block structured models have yet to be established.

 \item\label{prblm_mle}

Another open problem concerns the consistency properties of maximum likelihood estimators of the model parameters and latent factors under identifiability constraints. This question presents unique challenges, as the maximum likelihood estimator arises as the solution to a nonlinear and non-convex constrained optimization problem. Specifically, the identifiability constraints are intricately related to the block structure. The complex interaction between the block structure and the constraints introduces unique challenges and sets our analysis apart from existing work~\citep{bai2012statistical,chen2020structured,bing2020adaptive,Wang2018Maximum,cui2025identifiability}. Another challenge concerns the non-convexity of the log-likelihood induced by low-rank matrix factorization, even in the absence of block structures~\citep{burer2005local,sun2016guaranteed,cui2026convexity}.
Block structures further complicate the problem because different blocks affect the likelihood in distinct ways, with these effects depending on their sizes and configurations as well as the associated identifiability constraints.
Moreover, the considered framework allows nonlinear observation models, for which techniques tailored to linear models are no longer applicable~\citep{bai2003inferential,bai2012statistical,fan2017sufficient}. Finally, in many settings, both the number of observed variables and the sample can be large, leading to a high-dimensional estimation problem. The complex constraints, non-convexity, nonlinearity, and increasing dimensionality together make establishing consistency especially challenging.


\item\label{prblm_infe} 
A further related open problem concerns how to conduct statistical inference for the model parameters and latent factors. Under the block structure and identifiability constraints, each loading parameter characterizes the relationship between an observed variable and a common or block-specific latent dimension, and the latent factor reflects a subject-specific characteristic represented by that dimension. Inference for these quantities quantifies the associated uncertainty, providing statistical evidence for scientific discovery. 
Relatedly, \citet{cui2025identifiability} developed inferential theory for generalized latent factor models under a restricted class of identifiability constraints requiring minimal sparsity. Their analysis does not cover the present setting, in which the constraints interact intricately with the block structure, fundamentally altering the geometry of the nonconvex likelihood and thereby requiring new tools for estimation and inference.

\end{enumerate}

In this paper, we address all the problems outlined above: \ref{prblm_id} model identifiability, \ref{prblm_mle} consistency of the constrained maximum likelihood estimation, and \ref{prblm_infe} related statistical inference, for block structured latent variable models. Our main results are summarized as follows. First, to address problem~\ref{prblm_id}, we establish identifiability conditions for the block structured latent variable models. The proposed conditions are sufficiently flexible to accommodate a wide range of block structures commonly encountered in practice, including those considered in existing studies \citep{chen2020structured,fang2021identifiability,qiao2025exploratory}. See Examples~\ref{example_simple}--\ref{example_multi} in Section~\ref{subsec_setup} for more discussions.

Another main contribution is a novel Lagrangian-type reformulation of the constrained non-convex likelihood, developed to address problems~\ref{prblm_mle} and~\ref{prblm_infe}. Unlike traditional Lagrangian methods, it is designed so that its global optimum coincides exactly with that of the original constrained non-convex problem, without introducing any bias. 
Moreover, our analysis identifies the key blocks that determine the smallest eigenvalue of the Hessian of the Lagrangian function, thereby precisely characterizing how the block structure governs its local curvature. Based on this, we establish that, when the block structure is sufficiently informative, the resulting objective admits strong convexity in a neighborhood of the true parameters. We further show that this local minimizer is globally optimal and coincides with the constrained maximum likelihood estimator. Compared with the original problem, our Lagrangian-type formulation possesses more tractable properties, thereby enabling further statistical analysis of the maximum likelihood estimator.

Building on this Lagrangian-type formulation, we fully address problems~\ref{prblm_mle} and~\ref{prblm_infe} by developing a comprehensive statistical framework for the constrained maximum likelihood estimator. In particular, 
we establish non-asymptotic $\ell_2$- and $\ell_{\infty}$-error bounds (Theorem~\ref{thm_precise_consis}) and asymptotic distributions for the estimated latent factors and loading parameters (Theorem~\ref{thm_asymptotic_normality}).  
Across a wide range of asymptotic regimes, we show that the established $\ell_2$-error bounds for latent factors and loading parameters achieve the oracle rate. For the asymptotic distributions, we show that our established asymptotic variances are optimal in the sense of achieving the Cramer-Rao information lower bounds, enabling valid and efficient inference for both latent factors and loading parameters. Moreover, we develop an efficient first-order algorithm for computing the estimator under the nonlinear and non-convex constraints. We prove its linear convergence and show that, after sufficiently many iterations, the algorithmic output has the same asymptotic distributions as the exact estimator, bridging the gap between statistical theory and computation.

The rest of the paper is organized as follows. Section~\ref{sec:setup} introduces the block structured latent variable models and establishes their identifiability conditions.
Section~\ref{sec_estimation} examines the properties of the constrained maximum likelihood estimation viewed as an optimization problem. Section~\ref{sec_mainres} provides statistical guarantees for the estimator, including non-asymptotic error bounds and asymptotic normality. Section~\ref{sec_gdalgorithm} provides a first-order algorithm to obtain the estimator with theoretical guarantees. Section~\ref{sec:simu} uses simulation studies to illustrate the theoretical results.  We apply the proposed estimation framework to analyze an educational assessment dataset in Section~\ref{sec:data}. The paper concludes with final remarks in Section~\ref{sec:end}. The Supplementary Material presents all technical derivations, a computation framework and additional numerical results.

\paragraph*{Notation}
We use the following notation throughout. For two sequences $a_n$ and $b_n$, we denote $a_n\asymp b_n$ if their ratio is bounded above and below by positive constants.
For any integer $N$, let $[N] = \{1, \dots, N\}$. For a set $\cS$, let $|\cS|$ denote its cardinality. For two sets $\cS$ and $\cS^{\prime}$, let $\cS\cup\cS^{\prime}$ and $\cS\cap\cS^{\prime}$ denote their union and intersection, respectively, and let $\cS\setminus \cS^{\prime}$ denote set difference.
 For a matrix $\Ab \in \mathbb{R}^{n \times m}$ and index sets $\cS_1\subseteq[n]$, $\cS_2 \subseteq [m]$, let $\Ab_{[\cS_1,\cS_2]}\in\RR^{|\cS_1|\times|\cS_2|}$ denote the submatrix formed by the corresponding rows and columns. We abbreviate $\Ab_{[[n],\cS_2]}$ and $\Ab_{[\cS_1,[m]]}$ as $\Ab_{[,\cS_2]}$ and $\Ab_{[\cS_1,]}$, respectively. We use $\Ab_{[-\cS_1,-\cS_2]}$ for the submatrix obtained by removing these rows and columns.
For a vector $\bx\in\mathbb{R}^n$, define $\|\bx\| = (\sum_{i=1}^nx_i^2)^{1/2}$ and $\|\bx\|_{\infty} = \max_{1\le i\le n}|x_i|$. For a matrix $\Ab = (a_{ij})\in\RR^{n\times m}$ let $\ba_i$ denote its $i$th row and define $\| \Ab\|_{F} = (\sum_{i=1}^n\sum_{j=1}^ma_{ij}^2)^{1/2}$, $\| \Ab\|_{\max} = \max_{1\le i\le n,1\le j\le m} |a_{ij}|$, and $\|\Ab\|_{\infty}= \max_{i=1,\ldots,n} \sum_{j=1}^m |a_{ij}|$.
For any square matrix $\Ab\in \mathbb{R}^{n \times n}$, denote $\lambda_{\min}(\Ab)$ and $\lambda_{\max}(\Ab)$ as the smallest and largest eigenvalues of $\Ab$, respectively. Let $\Ib_n$ denote the $n \times n$ identity matrix, $\zero_n$ the $n$-dimensional zero vector, $\zero_{n\times m}$ the $n \times m$ zero matrix, and $\one_n$ the $n$-dimensional all-one vector.

\section{Block Structured Latent Variable Model}\label{sec:setup}
\subsection{Model Setup}\label{subsec_setup}
Consider the setting with $N$ individuals, each responding to $J$ items (manifest variables). For each individual $i\in[N]$, let $\bbf_i=(f_{i1},\cdots,f_{iD})^\T$ be a $D$-dimensional latent factor representing the individual's latent features, and let latent factor matrix be $\Fb = (\bbf_1,\dots,\bbf_N)^\T\in\RR^{N\times D}$. Here the number of latent dimensions $D$ is assumed to be known and fixed. 
The $J$ items are partitioned into $K$ non-overlapping blocks with index sets $\cJ_1,\dots,\cJ_K$, such that, for any $k\neq r$, $\cJ_k\cap \cJ_r= \varnothing$ and $\cup_{k=1}^K\cJ_k = [J]$. Within each block $k$, the item responses depend on a non-empty subset of the latent factors indexed by $\cA_k\subseteq[D]$. 
Specifically, for any item $j \in \cJ_k$, the response of individual $i$, denoted by $Y_{ij}$, follows a general distribution with density or mass function
\begin{equation}
    g_{ij}(Y_{ij} = y_{ij}\mid w_{ij(k)}),~\text{ where }~ w_{ij(k)} ~:= ~\beta_j + \sum_{d\in\cA_k}\lambda_{jd}f_{id},\label{eq_factor_m}
\end{equation}where $y_{ij}$ denotes the realization of $Y_{ij}$ and $g_{ij}(\cdot\mid\cdot)$ is a given function allowed to vary across $i$ and $j$.
Here, $\lambda_{jd}$ is the loading parameter that captures the influence of the $d$-th latent factor on item $j$, while $\beta_j$ is an intercept. The formulation \eqref{eq_factor_m} enables modeling different data types through suitably chosen function $g_{ij}(\cdot\mid\cdot)$. For instance, continuous responses can be modeled via $\log g_{ij}(y|w)\propto -(y-w)^2$ and binary responses can be modeled using either a logistic link $g_{ij}(y|w) = {\exp( w)^y}/{(1 + \exp( w))}$ or probit link $g_{ij}(y|w) = \Phi( w)^y (1 - \Phi( w))^{1-y}$ where $\Phi(\cdot)$ is the cumulative density function of standard normal distribution. We use a binary indicator matrix $\cQ = (q_{jd})\in\{0,1\}^{J\times D}$ to characterize the block structure, where $q_{jd} = 1$ if there exists some $k\in[K]$ such that $j\in\cJ_k$ and $d\in\cA_k$; and $q_{jd} = 0$ otherwise. 
Let $\bLambda_{\cQ}$ denote the loading matrix constrained by the block structure $\cQ$, whose $(j,d)$-th entry equals the free parameter $\lambda_{jd}$ if $q_{jd} = 1$ and equals zero otherwise. Moreover, let $\bbeta = (\beta_1,\dots,\beta_J)^\T$. The latent factors $\Fb$, loading parameters $\bLambda_{\cQ}$, and intercepts $\bbeta$ are all treated as model parameters of \eqref{eq_factor_m}, denoted collectively by $(\Fb,\bLambda_{\cQ},\bbeta)$.

Following common setups in the literature~\citep{andersonintroduction,skrondal2004generalized}, we assume responses $\{Y_{ij}\}_{i\in[N],j\in[J]}$ are conditionally independent given $\big\{w_{ij(k)}\}_{i\in[N],j\in\cJ_k,k\in[K]}$. To fix the location and scale of the latent factors, we adopt the normalization \begin{equation}N^{-1}\sum_{i=1}^Nf_{id} = 0~\text{ and }~N^{-1}\sum_{i=1}^Nf_{id}^2 = 1,\;\text{ for each }d\in[D].\label{eq_cond_normal}\end{equation} In addition, the loading matrix in block $k$, denoted by $\bLambda_k := \big(\lambda_{jd}\big)_{j\in \cJ_k,d\in\cA_k}$, is assumed to have full rank, that is, $\mathrm{rank}(\bLambda_k) = |\cA_k|$ for every $k\in[K]$.

Besides the block structure, another commonly imposed constraint is orthogonality among certain latent factors to ensure that they are uncorrelated, motivated by practical considerations. For example, in exploratory factor analysis, all factors are assumed to be orthogonal~\citep{bartholomew2011latent,chen2019joint}; in data integration, common factors are assumed to be orthogonal to source-specific ones~\citep{lock2013joint,feng2018angle}. These constraints reflect the intended interpretation of the latent dimensions in applications.
We encode these constraints using a symmetric binary indicator matrix $\cM = (m_{rl})_{D\times D}$ as follows. For any $r,l\in[D]$,
\begin{equation}
         m_{rl} ~=~\left\{\begin{aligned}&0\text{, if }\Fb_{[,r]}\text{ and }\Fb_{[,l]}\text{ are constrained to be orthogonal, i.e., } N^{-1}{\Fb_{[,r]}}^\T{\Fb_{[,l]}}=0;\\&1 \text{, otherwise.}\end{aligned}\right.\label{eq_cond2}
    \end{equation}
For instance, in exploratory factor analysis, $\cM = \Ib_D$; in the data integration setting, $\cM_{[\cS_1, \cS_2]} =\cM_{[\cS_2, \cS_1]}^\T = \zero_{|\cS_1|\times |\cS_2|}$, where $\cS_1$ and $\cS_2$ denote the index sets of common and source-specific latent factors, respectively.

In the following, we present four commonly used examples of block structure designs, with their graphical illustrations shown in Figure~\ref{fig:illus}. 

{\spacingset{1.19}\begin{figure}[h]
\centering    
        \includegraphics[width=2.7in]{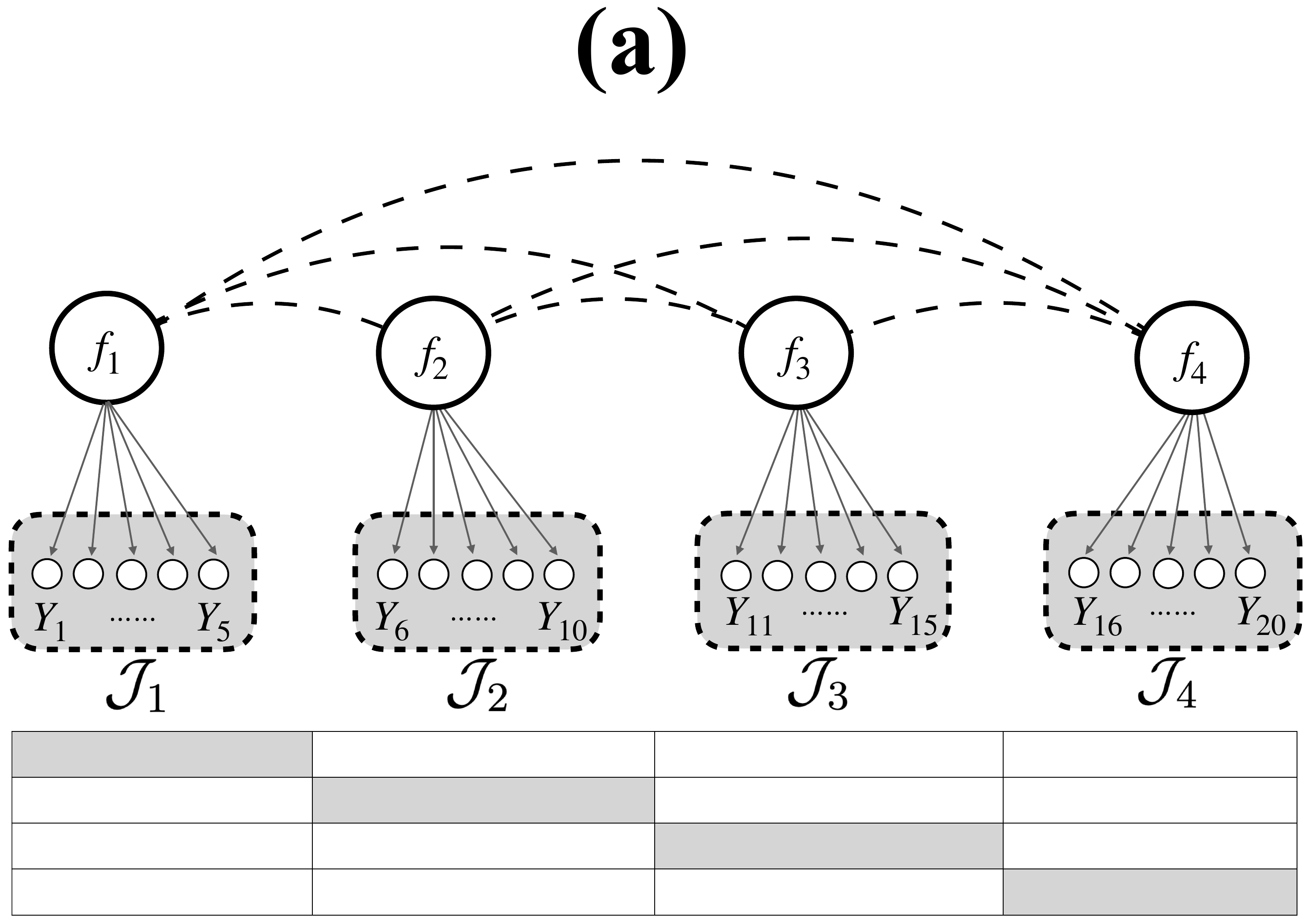}
        \includegraphics[width=2.7in]{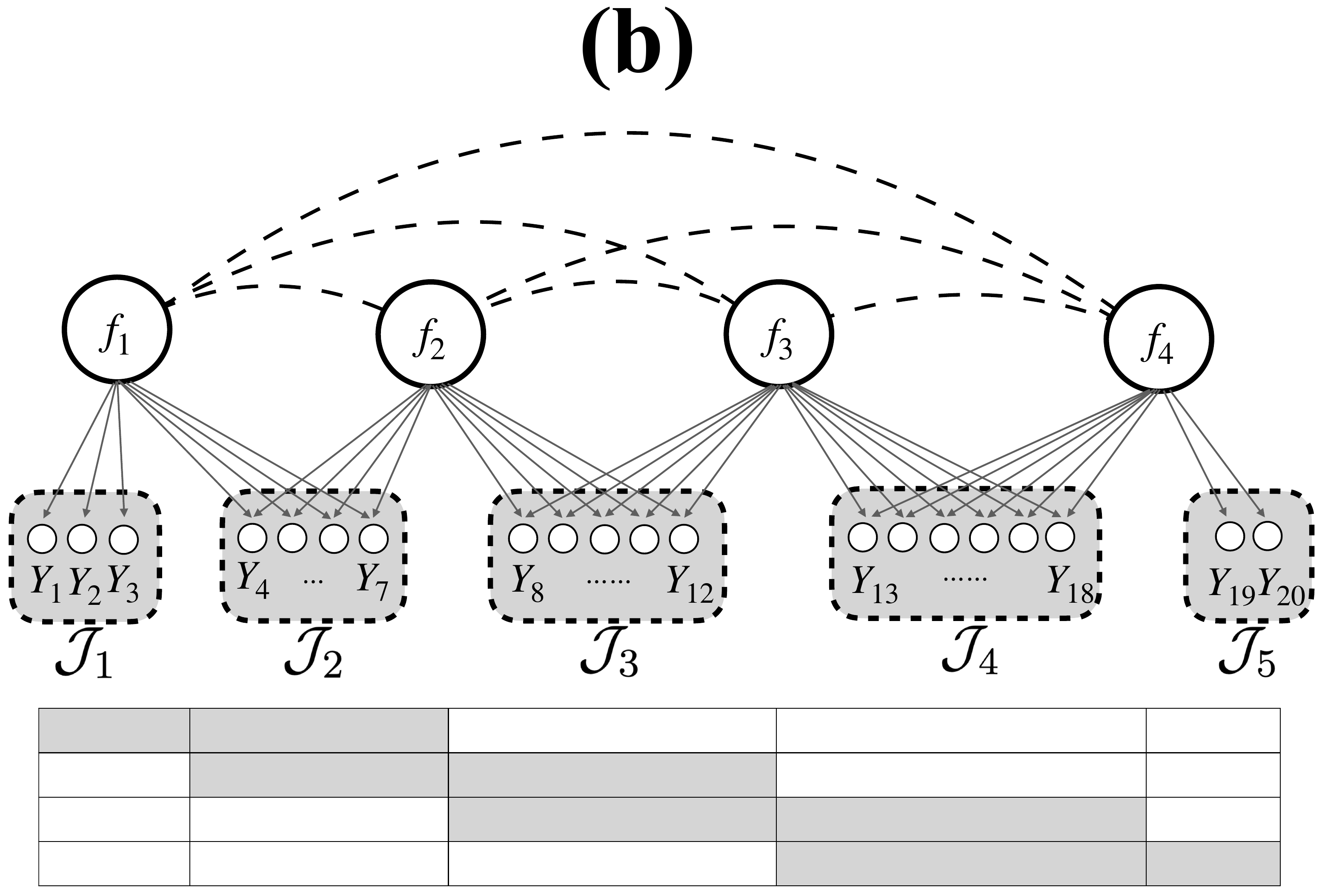}\\
        \includegraphics[width=2.7in]{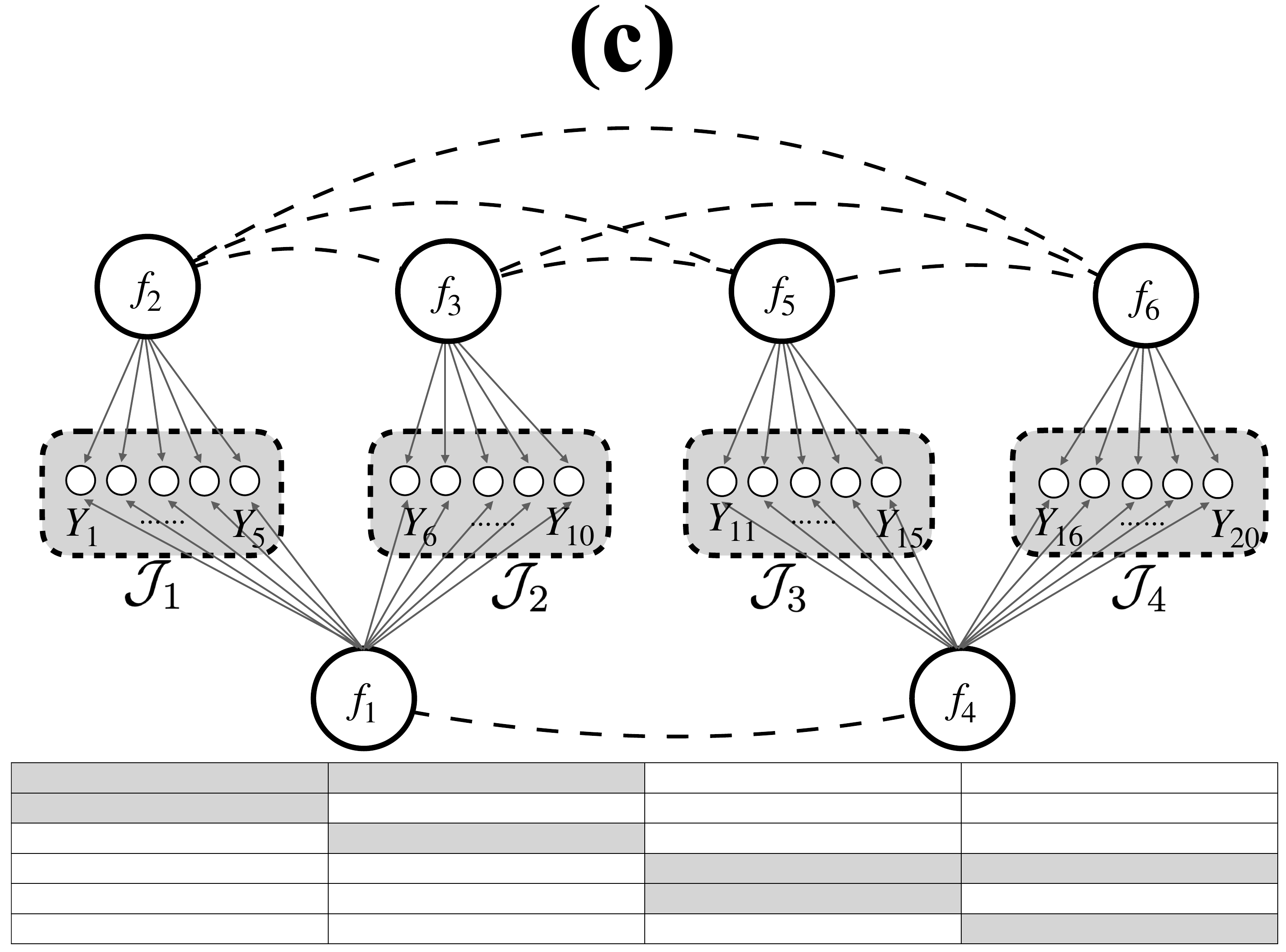}
        \includegraphics[width=2.7in]{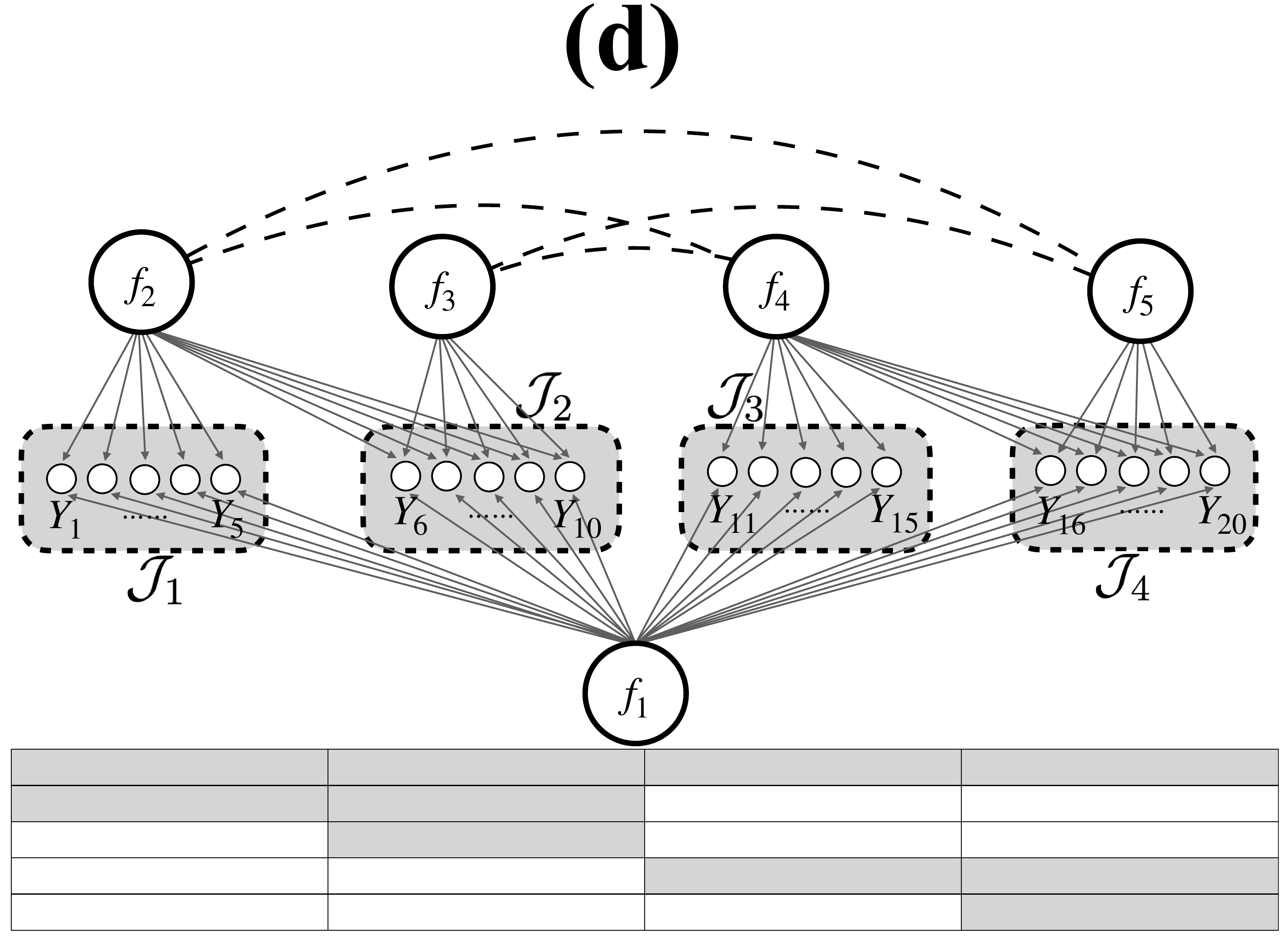}
        \caption{{\small Four examples of block structure designs in (a)--(d), with graphical representations shown in the upper panel and the corresponding matrix $\cQ^\T$ in the lower panel. In each graphical representation, dashed boxes denote blocks, and the nodes within them correspond to item response variables. Directed edges from latent factors to response variables indicate influence (absence of an edge implies independence). Undirected dashed edges among latent factors represent potential correlations. In $\cQ^\T$, black cells correspond to entries equal to $1$s, and white cells correspond to entries equal to $0$s.}}
    \label{fig:illus}
\end{figure}
}

\begin{example}\label{example_simple}\it One of the most popular block designs is the simple structure~\citep{reckase2009,trendafilov2014simple}, where the responses to the items in the same block are assumed to be influenced by one latent factor. This structure enhances factor differentiation and has been widely adopted for its interpretability, such as in Big-Five personality measurement~\citep{johnson2014measuring}. Figure~\ref{fig:illus}(a) provides one graphical illustration, where there are four latent factors $(f_{ 1}$--$f_{ 4})$ and four blocks $(\cJ_1$--$\cJ_4)$. For each $k\in[4]$, the items in block $\cJ_k$ are related to latent factor $f_{ k}$ only. 
In this setting, the factors are allowed to exhibit arbitrary correlations, which are often of substantive scientific interest~\citep{johnson2014measuring,chen2020structured}.
\end{example}

\begin{example}\it Moving beyond simple structures, block designs in which responses can be associated with multiple latent factors are also frequently encountered. These designs allow for more complex relationships between items' responses and latent factors, offering enhanced flexibility. It is worth mentioning that these structures have many applications in model-based clustering, especially in cases involving overlapping clusters~\citep{fraley2002model,bunea2020model,bing2020adaptive}. 
Figure~\ref{fig:illus}(b) shows an example with four clusters corresponding to four latent factors $(f_{ 1}$--$f_{ 4})$, where $\cJ_1\cup\cJ_2$ forms one cluster related to $f_{ 1}$ and $\cJ_2\cup\cJ_3$ forms another related to $f_{ 2}$, overlapping at $\cJ_2$ 
\end{example}

\begin{example}\label{example_bi_factor} \it  Bifactor model is another class of widely used block structured latent variable models~\citep{cai2010two,rodriguez2016evaluating,fang2021identifiability}. In this framework, latent factors are categorized into general factors, which capture pervasive constructs across multiple blocks, and domain-specific factors, which account for group-specific variation within a particular block. 
This modeling approach is particularly useful in psychological and educational assessments~\citep{cai2010two}, where it helps disentangle shared variance from specific group-level influences. Figure~\ref{fig:illus}(c) provides an illustration, where $f_{ 1}$ serves as the general factor for items in $\cJ_1$ and $\cJ_2$, $f_{ 2}$ acts as the general factor for items in $\cJ_3$ and $\cJ_4$, and $f_{ 3}$–$f_{ 6}$ represent domain-specific factors. In this setup, orthogonality constraints are often imposed between the general factors and domain-specific ones~\citep{cai2010two,fang2021identifiability}.
\end{example}

\begin{example}\label{example_multi}\it
Block structured latent variable models have also been applied in data integration~\citep{lock2013joint,feng2018angle,de2019multi}. Different from the bi-factor models, these models allow items' responses to be associated with more than two factors, resulting in a more complex structure. This framework is particularly effective for extracting common features from multiple data sources while distinguishing source-specific variations. One setup is provided in Figure~\ref{fig:illus}(d), which shows how data from two different sources can be integrated.  
Specifically, items in $\cJ_1$ and $\cJ_2$ come from one source, and items in $\cJ_3$ and $\cJ_4$ come from another, and both sources share a common factor $f_{ 1}$. Moreover, within items in $\cJ_1$ and $\cJ_2$, the subsets $\cJ_1$ and $\cJ_2$ can be viewed as arising from two sub-sources that share a factor $f_{ 2}$. A similar structure is exhibited by items from $\cJ_3$ and $\cJ_4$. Similar to Example~\ref{example_bi_factor}, orthogonality constraints between the general factors and domain-specific ones are often imposed in practice~\citep{lock2013joint, feng2018angle}.
\end{example}

\subsection{Model Identifiability}\label{sec_str_factor_model}
Before discussing our estimation method, we first study the model identifiability. A key indeterminacy in our model is that distinct parameter sets $(\Fb,\bLambda_{\cQ},\bbeta)$ and $(\bar\Fb,\bar\bLambda_{\cQ},\bar\bbeta)$ may satisfy $w_{ij(k)} = \beta_j + \sum_{d\in\cA_k}\lambda_{jd}f_{id}  = \bar\beta_j + \sum_{d\in\cA_k}\bar\lambda_{jd}\bar f_{id}$ for every $i\in[N]$, $j\in\cJ_k$, and $k\in[K]$, thereby leading to the same distribution for $Y_{ij}$. We formalize this specific identifiability issue as follows.
\begin{definition}[Structural Identifiability]\label{def_id}
Consider the model setup in Section~\ref{subsec_setup} with given block structure $\cQ$ and orthogonality constraint $\cM$. We say parameters $(\Fb,\bLambda_{\cQ},\bbeta)$ are structurally identifiable if the following holds: for any two parameter sets $(\Fb,\bLambda_{\cQ},\bbeta)$ and $(\bar\Fb,\bar\bLambda_{\cQ},\bar\bbeta)$ satisfying $\cQ$ and $\cM$, 
if \begin{equation*}\beta_j + \sum_{d\in\cA_k}\lambda_{jd}f_{id} ~ = ~ \bar\beta_j + \sum_{d\in\cA_k}\bar\lambda_{jd}\bar f_{id}\;\;\text{ for all }i\in[N],\;j\in\cJ_k\text{, and }k\in[K],\end{equation*} then $(\Fb,\bLambda_{\cQ},\bbeta)$ and $(\bar\Fb,\bar\bLambda_{\cQ},\bar\bbeta)$ must be equivalent up to a sign change of the latent factors and loading parameters, that is, there exists some diagonal matrix $\Gb$ with diagonal entries in $\{\pm 1\}$ such that $\bar\bLambda_{\cQ} = \bLambda_{\cQ}\Gb$, $\bar\Fb=\Fb\Gb$, and $\bar\bbeta = \bbeta$. 

\end{definition}

In Definition~\ref{def_id}, we permit column sign changes on the latent factors and loading parameters. This indeterminacy can be readily resolved by specifying that, for certain items, the latent factors exert a positive or negative influence on their responses. These specifications are often determined based on the interpretations of the latent factors.
\begin{remark}[Relation to Full Model Identifiability] \label{rem_full_id} 
Let $\mathbb P_{(\Fb,\bLambda_{\cQ},\bbeta)}$ denote the distribution of all observed responses under the considered model with $(\Fb,\bLambda_{\cQ},\bbeta)$. Full model identifiability requires
\begin{equation*}
\mathbb P_{(\Fb,\bLambda_{\cQ},\bbeta)} = \mathbb P_{(\bar\Fb,\bar\bLambda_{\cQ},\bar\bbeta)}\quad\implies\quad
(\Fb,\bLambda_{\cQ},\bbeta) =(\bar\Fb,\bar\bLambda_{\cQ},\bar\bbeta).\end{equation*}
Given structural identifiability as in Definition~\ref{def_id}, this property follows if the response distribution is identifiable from $w$, in the sense that
\begin{equation}
g_{ij}(\cdot\mid w)=g_{ij}(\cdot\mid \bar w) \quad\implies\quad w=\bar w \label{eq_response_identifiability}
\end{equation}
for every $i\in[N]$ and $j\in[J]$. Condition \eqref{eq_response_identifiability} can be ensured by the concavity imposed in Assumption~\ref{assumption_smoothness} below; see Section~\ref{supp_sec_model_id} of the Supplementary Material for verification. Hence, structural identifiability, together with \eqref{eq_response_identifiability} and the sign conventions described above, yields full model identifiability. Our analysis thus focuses on structural identifiability as it requires more delicate analysis and informs model interpretation, consistent with related treatments of nonlinear latent factor models~\citep{chen2020structured,cui2025identifiability}.
\end{remark}

Our goal is to determine, given an arbitrary pair $(\cM,\cQ)$, whether the parameters satisfying these constraints are structurally identifiable in the sense of Definition~\ref{def_id}. This problem is challenging under the joint presence of $\cM$, which involves nonlinear relations among the latent factors, and $\cQ$, which imposes entry-wise restrictions on the loading parameters. These constraints together define a nonlinear parameter space, within which establishing structural identifiability is highly nontrivial. We illustrate this issue via the following example.

\begin{example}\label{example_non_id}\it
    Consider the setup in Figure~\ref{fig:illus}(d) with $K=4$, $D=5$, and $J=20$. The block structure $\cQ$ is specified as in the figure, and the orthogonality constraint $\cM$ is \begin{equation}
    \cM=\begingroup\setlength{\arraycolsep}{5pt}
\renewcommand{\arraystretch}{0.4}
\left(
\begin{array}{c:c c c c} 
1 & 0 & 0 & 0 & 0 \\\hdashline  
0 & 1 & 0 & 1 & 1 \\
0 & 0 & 1 & 1 & 1 \\
0 & 1 & 1 & 1 & 0 \\
0 & 1 & 1 & 0 & 1
\end{array}
\right)
\endgroup.\label{eq_exampleof_cM}
\end{equation}Under this setup, the model parameters are structurally identifiable, as shown by Theorem~\ref{thm_id_iff} (see Example~\ref{example_id}).
However, relaxing any of the orthogonality constraints in $\cM$ makes the model parameters no longer structurally identifiable.
Consider an alternative orthogonality constraint $\cM^{\prime}$ defined by $\cM^{\prime} = \cM$ except for $\cM^{\prime}_{[2,3]} = \cM^{\prime}_{[3,2]} = 1$. Let $(\Fb,\bLambda_{\cQ},\bbeta)$ denote the model parameters. Let $\rho_{23} = N^{-1}\Fb_{[,2]}{}^\T\Fb_{[,3]}$. Given any scalars $x,y$ satisfying $x^2 + 2xy\rho_{23}+ y^2  =1$, define $(\bar\Fb,\bar\bLambda_\cQ,\bar\bbeta)$ as follows:
    \begin{equation*}
        \text{ for }j\in\cJ_2\cup\cJ_3,\; \bar\lambda_{j2} = \lambda_{j2} - \frac{y}{x}\lambda_{j3};\;\text{ for }j\in\cJ_3\text{, }\bar\lambda_{j3} = \frac{1}{x}\lambda_{j3};\text{ for }i\in[N],\;\bar f_{i3} =  xf_{i3}  + yf_{i2};
    \end{equation*}
    and all other parameters are set equal to those in $(\Fb,\bLambda_{\cQ},\bbeta)$. It can be checked that $(i)$ $\beta_j + \sum_{d\in\cA_k}\lambda_{jd}f_{id}  = \bar\beta_j + \sum_{d\in\cA_k}\bar\lambda_{jd}\bar f_{id}~$ for every $i\in[N]$, $j\in\cJ_k$, and $k\in[K]$; $(ii)$ parameters $(\bar\Fb,\bar\bLambda_\cQ,\bar\bbeta)$ satisfy $\cQ$ and $\cM^{\prime}$, as well as the constraints in \eqref{eq_cond_normal}. Thus, the parameters are not structurally identifiable under the considered  $\cM^{\prime}$ and $\cQ$.
    \end{example}

Example~\ref{example_non_id} illustrates that not all combinations of orthogonality constraint $\cM$ and block structure $\cQ$ guarantee structural identifiability, and the resulting indeterminacy is case-specific.
To determine structural identifiability under general $\cM$ and $\cQ$, we introduce a novel condition that implicitly confines $\cM$ and $\cQ$. 
We begin with the following definition. 
\begin{definition}
    Let $\mathscr{A}$ be a collection of subsets of $[D]$ and $\cM\in\{0,1\}^{D\times D}$. The $\cM$-extended $\pi$-system generated by $\mathscr{A}$, denoted by $\pi_{\cM}(\mathscr{A})$, is defined as the smallest collection satisfying the following conditions:
    \begin{enumerate}[label=$(\roman*)$]
    \item $\mathscr{A}\subseteq \pi_\cM(\mathscr{A})$;
    \item {For any positive integer $n$ and any $\{\cS_i\}_{i=1}^n \subseteq \pi_\cM(\mathscr{A})$,  
    $\cap_{i=1}^n \cS_i\in \pi_\cM(\mathscr{A})$. Furthermore, if  $\cS_0\in\piSyst$ satisfies $\cM_{[\cup_{r=1}^n \cS_r,\, \cS_0 \setminus (\cup_{r=1}^n \cS_r)]} = 0$, then $\bigcap_{r=1}^n \big(\cS_0 \setminus \cS_r\big) \in \pi_\cM(\mathscr{A})$.}
\end{enumerate}
\end{definition}
This $\cM$-extended $\pi$-system serves as a key tool to link the orthogonality constraint $\cM$ and block structure $\cQ$. In particular, we introduce the following condition.

\vspace{0.2cm}
\noindent\textbf{$\cM$-$\cQ$ Condition.} For any $r\in[D]$, $\{r\}\in\pi_{\cM}(\mathscr{A}_{\cQ})$, where $\mathscr{A}_{\cQ} :=\{\cA_k\}_{k\in[K]}$.
\vspace{0.2cm}

The collection $\piSyst$ integrates information from both the orthogonality constraint $\cM$ and block structure $\cQ$ in a carefully designed manner, encoding important structural relationships to determine the identifiability. In particular, it is designed to contain index sets that are distinguishable under these constraints, and the requirement that every singleton set $\{r\}$ belongs to $\piSyst$ ensures sufficient discrepancies between latent factors.
One case excluded by this condition is when two latent factors share identical block memberships, i.e., they either both appear or both do not appear in every $\cA_{k}$ for each $k\in[K]$. In such a scenario, no set $\cA\in\piSyst$ could include one latent factor while excluding the other, regardless of the $\cM$ imposed, thereby violating the $\cMQ$ Condition. 

The next theorem shows the sufficiency and necessity of the $\cMQ$ Condition for structural identifiability.
\begin{theorem}[Identifiability]\label{thm_id_iff}\it
 Consider the model specified in Section~\ref{sec:setup} with orthogonality {constraint} $\cM$ and block structure $\cQ$. The corresponding model parameters $(\Fb,\bLambda_{\cQ},\bbeta)$ are structurally identifiable if and only if the $\cM$-$\cQ$ Condition holds.
\end{theorem}
Theorem~\ref{thm_id_iff} determines when a given block structure $\cQ$ and orthogonality constraint $\cM$ ensure structural identifiability. The theorem can also be applied in other practical scenarios: given a specified block structure $\cQ$, it allows practitioners to determine the minimal orthogonality constraint that should be imposed for identifiability; conversely, given pre-specified orthogonality constraint $\cM$ provided by practitioners, it offers a guideline for designing the block structures $\cQ$ under which the model parameters are identifiable. 
\begin{example}\label{example_id}
    \it
    The $\cMQ$ Condition is satisfied for setups (a)–(d) of Figure~\ref{fig:illus} and thus, by Theorem~\ref{thm_id_iff}, the model parameters are structurally identifiable for these setups. Below, we verify the $\cMQ$ Condition for case (d), and arguments for cases (a)--(c) are left to Section~\ref{sec_continue_remark_id_cM} 
    of the Supplementary Material.
    In Figure~\ref{fig:illus}(d), $\mathscr{A}_{\cQ} = \big\{\{1,2\},\{1,2,3\},\{1,4\},\{1,4,5\}\big\}$; $\cM$ is given in \eqref{eq_exampleof_cM}. We can then verify that each singleton belongs to $\piSyst$: $\{1\} = \{1,2\}\cap \{1,4\}$; $\{2\} = \{1,2\}\setminus\{1\}$ with $\cM_{[1,2]} = 0$; $\{3\} = \{1,2,3\}\setminus\{1,2\}$ with $\cM_{[\{1,2\},3]} = \zero_2$; and similarly for $\{4\},\{5\}\in\piSyst$. By contrast, for $\cM^{\prime}$ in Example~\ref{example_non_id}, we observe that $\{3\}\notin \pi_{\cM^{\prime}}(\mathscr{A}_{\cQ})$, and thus the $\cMQ$ Condition is not satisfied for this $\cM^{\prime}$ and $\cQ$ in Figure~\ref{fig:illus}(d). 
    
\end{example}
 
\begin{remark}\label{remark_minimal_id}\it
Given $\cM$ and $\cQ$, we say that the $\cMQ$ Condition is satisfied minimally if the constraint $\cM$ cannot be further relaxed, i.e., changing any $m_{rl}$ from 0 to 1 would cause the $\cMQ$ Condition to fail, such as the situation discussed in Example~\ref{example_non_id}. To streamline our analysis and avoid redundant conditions, we focus on this minimal identifiable scenario and, for simplicity, refer to the $\cMQ$ Condition as being satisfied when it is satisfied in this minimal sense.
\end{remark} 


\section{Local Convexity and Global Optimality of Log-likelihood Function}\label{sec_estimation}

Now we study the maximum likelihood estimation under the orthogonality constraint $\cM$ and the block structure $\cQ$. 
The negative log-likelihood function as a loss is given as
\begin{equation*}
    L\big(\Fb,\bLambda_{\cQ},\bbeta\mid\Yb\big) ~=~-\sum_{i\in[N]}\sum_{k\in[K]}\sum_{j\in\cJ_k} l_{ij}\Big(\beta_j + \sum_{d\in\cA_k}\lambda_{jd}f_{id}\Big),
\end{equation*}
where $l_{ij}(\cdot) := \log g_{ij}(Y_{ij}\mid \cdot)$. 
Subsequently, the maximum likelihood estimator is given as
\begin{align}
\big(\hat\Fb,\hat\bLambda_{\cQ},\hat\bbeta\big)~ = \mathop{\arg\min}_{(\Fb,\bLambda_{\cQ},\bbeta)\in \Xi_{\cM\text{-}\cQ}(M)} L\left(\Fb,\bLambda_{\cQ},\bbeta\mid\Yb\right). \label{eq_constrained_maximum_likelihood_estimation} 
\end{align}
Here, the feasible set $\Xi_{\cMQ}(M)$ defines a bounded region in which the parameters satisfy block structure $\cQ$, the orthogonality constraint $\cM$, and \eqref{eq_cond_normal}. Specifically, let\begin{equation}\begin{aligned}
    \Xi_{\cM\text{-}\cQ} (M) ~=~ \Big\{(\Fb,\bLambda_{\cQ},\bbeta)\;:&\;\|\Fb\|_{\max}\le M,\;\|\bLambda_{\cQ}\|_{\max}\le M,\;\|\bbeta\|_{\infty}\le M,\\&\,\Fb\text{ satisfy constraint }\cM\text{ and }\eqref{eq_cond_normal} \Big\}\end{aligned}\label{eq_regime_true}
\end{equation}
 for some $M$ that is selected large enough to regulate the
magnitudes of the parameters. 


One difficulty in analyzing \eqref{eq_constrained_maximum_likelihood_estimation} is that, as $J\to\infty$, different blocks may diverge at different rates, exerting distinct effects on the curvature of the log-likelihood function. Later in Theorem~\ref{thm_suff_necce_concavity}, we show that the curvature is determined solely by specific blocks that are critical for satisfying the $\cMQ$ Condition. 



Another notable challenge arises from the constraints in $\Xi_{\cMQ}(M)$. In contrast to the standard orthogonality constraints $N^{-1}\Fb^\T\Fb = \Ib_D$ adopted in existing literature~\citep{bai2012statistical,fan2017sufficient,chen2019joint,cui2025identifiability}, the constraint $\cM$ imposes $N^{-1}{\Fb_{[,r]}}^\T\Fb_{[,l]}=0$ for selected pairs $(r,l)$. 
Whereas $N^{-1}\Fb^\T\Fb = \Ib_D$ restricts $\Fb$ on a Stiefel manifold~\citep{plumbley2005geometrical}, the constraints considered in our setting lead to a manifold with no convenient structure, which makes direct analysis of the optimization problem \eqref{eq_constrained_maximum_likelihood_estimation} nearly infeasible.
To overcome this difficulty, we construct a Lagrangian-type formulation of \eqref{eq_constrained_maximum_likelihood_estimation}, which we later show yields the same optimum as the original problem. Specifically, we introduce a quadratic penalty function as follows: 
\begin{equation*}
    P(\Fb) ~=~ \frac{1}{2}NJ\Big\{\big\|N^{-1}\Fb^\T\one_N\big\|^2 + \big\|\mathrm{diag}(\Mb_{ff})-\Ib_D\big\|_F^2 +\big\|\Mb_{ff} - \Mb_{ff}\circ \cM\big\|_F^2 \Big\},
\end{equation*}
where $\Mb_{ff}:=N^{-1}\Fb^\T\Fb$, $\mathrm{diag}(\Mb_{ff})$ is a diagonal matrix preserving only the diagonal elements in $\Mb_{ff}$, and $\circ$ represents the element-wise (Hadamard) product.
Note that $P(\Fb)=0$ precisely when latent factors $\Fb$ satisfy constraints in \eqref{eq_cond_normal} and $\cM$.  
The term $NJ/2$ serves as a scaling factor to align with the scale of the log-likelihood function. With function $P(\Fb)$, we let $\cL_{\nu}(\Fb,\bLambda_{\cQ},\bbeta) := L(\Fb,\bLambda_{\cQ},\bbeta) + \nu P(\Fb)$ and construct the Lagrangian-type formulation of \eqref{eq_constrained_maximum_likelihood_estimation} as
\begin{equation}
     (\tilde\Fb,\tilde\bLambda_{\cQ},\tilde\bbeta) ~= \mathop{\arg\min}_{\|\Fb\|_{\max}\le M,\|\bLambda_{\cQ}\|_{\max}\le M,\|\bbeta\|_{\infty}\le M} \cL_{\nu}(\Fb,\bLambda_{\cQ},\bbeta),
     \label{key_optimization_problem}
\end{equation} where $M$ is the same as in \eqref{eq_constrained_maximum_likelihood_estimation}. For brevity, we sometimes write $\cL_{\nu}(\cdot)$ without explicitly listing its arguments. 
\begin{remark}
    \it While quadratic penalty functions have been used to construct Lagrangian-type formulations in previous studies~\citep{chen2021nonlinear,Wang2018Maximum,li2023statistical,wu2024general}, the analysis here faces some unique challenges. Specifically, existing works impose conditions on the entire matrix $N^{-1}\Fb^\T\Fb$, which relates to restrictions on the Stiefel manifold and thus yield a simpler Hessian structure. In contrast, our condition involves $N^{-1}{\Fb_{[,r]}}^\T\Fb_{[,l]}=0$ for selected pairs $(r,l)$. Moreover, the selected pairs interact intricately with the given block structure, as confined by the $\cMQ$ Condition. This interaction is the central difficulty unique to our framework, requiring new technical tools. 
\end{remark}

This Lagrangian-type formulation possesses some useful properties that help study the original constrained problem \eqref{eq_constrained_maximum_likelihood_estimation}. We give an informal summary of its properties below, with further details provided in the following subsections. 
\begin{enumerate}[label=$(\roman*)$]
    \item \label{res_convex}(Local Strong Convexity) $\cL_{\nu}(\cdot)$ is strongly convex within a neighborhood of the true parameters. Specifically, let $\lambda_{\min,\nu}(\Fb,\bLambda_{\cQ},\bbeta)$ denote the smallest eigenvalue of the properly scaled Hessian of $\cL_{\nu}(\cdot) $(definition given in \eqref{eq_scaled_minimal_eigenvalue}). There exist constants $\gamma>0$ and $C>0$ such that, when $N$ and $J$ are large enough, for any constant $c>0$, with high probability,
    \begin{equation*}
    \min_{(\Fb,\bLambda_{\cQ},\bbeta)\in\Xi_{\cQ}^*(c,M)}\lambda_{\min,\nu}\big(\Fb,\bLambda_{\cQ},\bbeta\big)~\ge~\gamma{J_*}/{J} - C\Big\{\nu\sqrt{c} + \sqrt{1/(N\wedge J)}\Big\},
\end{equation*} where $J_*$ denotes the size of the smallest non-removable block for the $\cMQ$ Condition to hold (definition given in \eqref{eq_spec_J_star}), $\Xi_{\cQ}^*(c,M)$ is a neighborhood region around the true parameters defined below in \eqref{eq_local_regime_true}, and $N\wedge J = \min\{N,J\}$. Moreover, this lower bound is tight in terms of the ratio $J_*/J$: there exists some true parameter set $(\Fb^*,\bLambda_{\cQ}^*,\bbeta^*)\in\Xi_{\cMQ}(M)$ such that, with high probability,
\begin{equation*}
        \lambda_{\min,\nu}\big(\Fb^*,\bLambda_{\cQ}^*,\bbeta^*\big)~\asymp ~\Big\{J_*/J + \sqrt{1/(N\wedge J)}~\Big\}\text{ for any constant }\nu>0;
    \end{equation*}
\item\label{res_equi} (Global Optimality) solutions of \eqref{eq_constrained_maximum_likelihood_estimation} and \eqref{key_optimization_problem} are equivalent, both coinciding with the unique minimizer of $\cL_{\nu}(\cdot)$ in $\Xi_{\cQ}^*(c,M)$ for any constant $c>0$. Specifically, for any constants $c>0$ and $\nu>0$, when $N$ and $J$ are large enough, with high probability,
\begin{equation*}
        (\hat\Fb,\hat\bLambda_{\cQ},\hat\bbeta) ~=~ (\tilde\Fb,\tilde\bLambda_{\cQ},\tilde\bbeta) ~= \mathop{\arg\min}_{(\Fb,\bLambda_{\cQ},\bbeta)\in\Xi_{\cQ}^* (c,M)} \cL_{\nu}(\Fb,\bLambda_{\cQ},\bbeta).
    \end{equation*}
\end{enumerate}
Here, $\Xi_{\cQ}^* (c,M)$ is a local region around the true parameters given by
\begin{align}
    \Xi_{\cQ}^* (c,M) ~=~ \Big\{(\Fb,\bLambda_{\cQ},\bbeta):&\,\;\|\Fb\|_{\max}\le M,\|\bLambda_{\cQ}\|_{\max}\le M,\|\bbeta\|_{\infty}\le M,\nonumber\\&\,\frac{1}{N}\|\Fb - \Fb^*\|_F^2 + \frac{1}{J}\|\bLambda_{\cQ}-\bLambda_{\cQ}^*\|_F^2+\frac{1}{J}\|\bbeta-\bbeta^*\|^2\le c\Big\},\label{eq_local_regime_true}
\end{align}
where $c$ is a small constant confining the size of this local region, $M$ is selected the same as in \eqref{eq_regime_true}, and $(\Fb^*, \bLambda_{\cQ}^*, \bbeta^*)$ are the true parameters.  Since the model parameters are identifiable up to a sign change, the result holds for the estimators after applying appropriate sign adjustments to both the loadings and the latent factors. For conciseness, we omit these adjustments here and in all subsequent results.


Our results offer three key implications. 
\begin{enumerate}[label=$\bullet$,leftmargin=0.5cm]
    \item  First and most importantly, result~\ref{res_equi} establishes that the local optimum of $\cL_{\nu}(\cdot)$ is globally optimal and coincides with the constrained estimator $(\hat\Fb,\hat\bLambda_{\cQ},\hat\bbeta)$ given in \eqref{eq_constrained_maximum_likelihood_estimation}. Thus, the theoretical properties of $(\hat\Fb,\hat\bLambda_{\cQ},\hat\bbeta)$ can be derived by analysing the more tractable Lagrangian-type formulation \eqref{key_optimization_problem} within the local region $\Xi_{\cQ}^*(c,M)$.

\item Second, the minimizer shared by problems \eqref{eq_constrained_maximum_likelihood_estimation} and \eqref{key_optimization_problem} is unique,  provided the ratio $J_*/J$ is bounded away from zero. Specifically, result~\ref{res_convex} shows that the Lagrangian function $\cL_{\nu}(\cdot)$ is strongly convex in $\Xi_{\cQ}^*(c,M)$ for some small constant $c>0$ when $J_*/J$ is non-diminishing. Combined with result~\ref{res_equi}, this strong convexity ensures the uniqueness of the solution to both optimization problems and guarantees that it lies within $\Xi_{\cQ}^*(c,M)$. 


\item The third implication concerns the behavior of the log-likelihood, where we identify $J_*$ as a key quantity controlling the convexity of $\cL_{\nu}(\cdot)$. Specifically, $J_*$ represents the size of the smallest non-removable blocks needed for the $\cMQ$ Condition to hold and is given by\begin{equation}
    J_* ~:=~ \max_{\cS\in\cD}\min_{k\in \cS}|\cJ_k|,\label{eq_spec_J_star}
\end{equation}with $\cD:=\big\{\cS\subseteq[K]\text{:}\text{ the $\cMQ$ Condition is satisfied with }\mathscr{A}_{\cQ}\text{ replaced by }\{\cA_k\}_{k\in\cS}\big\}$. As implied by result~\ref{res_convex}, the ratio $J_*/J$ characterizes the scaled minimal eigenvalues. Intuitively, since $J_*$ corresponds to the minimal block size for the $\cMQ$ Condition to hold, keeping $J_*/J$ bounded away from zero preserves enough discrepancies between the factors asymptotically, which helps maintain the curvature of the objective function and ensures the convexity. 

\end{enumerate}




\subsection{Local Convexity}\label{subsec_local_conv}
In this subsection, we examine the behavior of $\cL_{\nu}(\cdot)$ in the local region $\Xi_{\cQ}^*(c,M)$.
Following the literature, we introduce two commonly used regularity conditions. The true parameters are $(\Fb^*,\bLambda^*_{\cQ},\bbeta^*)$. 
For $k\in[K]$, let $\bLambda_k^*=(\lambda_{jd}^*)_{j\in\cJ_k,d\in\cA_k}$ be the true loading matrix for the $k$-th block. Let $J_k = |\cJ_k|$ be the number of items in block $k$. For $i\in[N]$, $k\in[K]$ and $j\in\cJ_k$, define $w_{ij(k)}^{*} = \sum_{d\in\cA_k}\lambda_{jd}^*f_{id}^* + \beta_j^*$. 
\begin{assumption}
	\label{assumption_psd_covariance}  \begin{enumerate}[label=$(\roman*)$]
        \item  There exists a constant $M_1$ such that $(\Fb^*,\bLambda^*_{\cQ},\bbeta^*)\in\Xi_{\cMQ}(M_1)$.
        \item $\bSigma_f^* = \lim\limits_{N \rightarrow \infty}N^{-1}{\Fb^*}^\T\Fb^* $ and $\bSigma_{\lambda}^* = \lim\limits_{J \rightarrow \infty}J^{-1}{\bLambda^*_{\cQ}}^\T\bLambda^*_{\cQ} $ exist and are positive definite.
        \item The eigenvalues of $\bSigma_{f}^*\bSigma_{\lambda}^*$ are distinct and nonzero.
        \item${\liminf}_{J\to\infty}\lambda_{\min}(J_k^{-1}\bLambda_k^*{}^\T\bLambda_k^*) > 0$ for all $k\in[K]$.
    \end{enumerate}
    
\end{assumption}\begin{remark}\it Assumption~\ref{assumption_psd_covariance} is commonly used in the latent factor literature. Specifically,  Assumption~\ref{assumption_psd_covariance}(i) requires the true parameters to be bounded and to satisfy the orthogonal constraint $\cM$ and the block structure $\cQ$. 
Assumption~\ref{assumption_psd_covariance}(ii) 
and Assumption~\ref{assumption_psd_covariance}(iii) are standard regularity conditions in the literature~\citep{stock2002forecasting,bai2003inferential,fan2017sufficient,fan2023bridging}.
Assumption~\ref{assumption_psd_covariance}(iv) requires that for each block $k$, the smallest eigenvalue of $J_k^{-1}\bLambda_k^*{}^\T\bLambda_k^*$ remain bounded away from zero. This assumption ensures that the items in each block $k$ contain sufficient information on the latent factors in $\cA_k$~\citep{chen2020structured}.
\end{remark}

 	\begin{assumption}
	\label{assumption_smoothness}
  The function $l_{ij}( w)  := \log g_{ij}(Y_{ij}\mid w)$ is three times differentiable to $ w$, with the first, second, and third order derivatives denoted by $l_{ij}^{\prime}( w)$, $l_{ij}^{\prime\prime}( w)$, and $l_{ij}^{\prime\prime\prime}( w)$, respectively. For any $i\in[N]$ and $j\in[J]$,  $l_{ij}^{\prime} ( w_{ij(k)}^*)$ is sub-exponential  with sub-exponential norm $\|l_{ij}^{\prime} ( w_{ij(k)}^*)\|_{\varphi_1} \le M_2$ for some constant $M_2>0$. 
 Within a compact set of $w$, there exist $0<b_L<b_U$ such that $b_L\le -l_{ij}^{\prime\prime}( w)\le b_U$ and $|l_{ij}^{\prime\prime\prime}( w)| \leqslant b_U$. 
   \end{assumption}

\begin{remark}
    \it
Assumption~\ref{assumption_smoothness} imposes smoothness conditions on the individual log-likelihood function $l_{ij}(w)$, which are standard in the literature~\citep{fernandez2016individual,Wang2018Maximum}. The sub-exponential tail assumption on $l_{ij}^{\prime}(w_{ij(k)}^*)$ is mild and holds for a broad class of commonly used link functions, including linear, logistic, probit, and Poisson. The conditions on the second and third order derivatives are technical regularity requirements for the theoretical analysis and are likewise satisfied by these links. 
\end{remark}

Now we formally introduce our results for the minimal eigenvalues of the scaled Hessian matrix of $\cL_{\nu}(\cdot)$, given as
\begin{equation}
    \lambda_{\min,\nu}(\Fb,\bLambda_{\cQ},\bbeta) ~=~ \lambda_{\min}\big\{\Sbb_m\partial^2_{\bpsi\bpsi}\cL_{\nu}(\Fb,\bLambda_{\cQ},\bbeta)\Sbb_m \big\}.\label{eq_scaled_minimal_eigenvalue}
\end{equation}
Here, $\bpsi := \big\{\bbf_1^\T,\dots,\bbf_N^\T,\mathrm{vec}(\bLambda_1)^\T,\dots,\mathrm{vec}(\bLambda_K)^\T,\beta_1,\dots,\beta_J\big\}^\T$ denotes the vector for the model parameters, where $\mathrm{vec}(\Ab)$ is the vectorization operator that stacks the columns of a matrix $\Ab$ into a vector; $\Sbb_m := \mathrm{diag}(J^{-1/2}\Ib_{ND},N^{-1/2}\Ib_{J_{\cQ}})$ is a scaling matrix where $J_\cQ = \sum_{k=1}^KJ_k\times |\cA_k| + J$ represents the number of free loading parameters and intercepts.

\begin{theorem}[Local strong convexity]\label{thm_suff_necce_concavity} \it
    Suppose Assumptions~\ref{assumption_psd_covariance}, \ref{assumption_smoothness} and the $\cMQ$ Condition hold. Let $\nu_0>0$ be a positive lower bound for the Lagrangian multiplier $\nu$. Then for any $\delta>1$, $\nu>\nu_0$ and $c>0$, there exists some constant $\gamma_{\delta}>0$ depending on $\delta$ such that, when $N$ and $J$ are sufficiently large, with probability at least $1-(N\wedge J)^{-\delta}$,
\begin{equation}
    \min_{(\Fb,\bLambda_{\cQ},\bbeta)\in\Xi_{\cQ}^*(c,M)}\lambda_{\min,\nu}(\Fb,\bLambda_{\cQ},\bbeta)~\ge~\gamma_{\delta}{J_*}/{J} - C_{\delta}\big\{\nu\sqrt{c} + \sqrt{1/({N\wedge J})}\big\},\label{eq_well_conditioning}
\end{equation}for some constant $C_{\delta}>0$ depending only on $\delta$.
On the other hand, there always exists a set of true parameters $(\Fb^*,\bLambda_{\cQ}^*,\bbeta^*)\in\Xi_{\cMQ}(M)$ such that, for any constant $\delta>1$, when $N$ and $J$ are large enough, with probability at least $1-(N\wedge J)^{-\delta}$,
\begin{equation}
    \lambda_{\min,\nu}(\Fb^*,\bLambda_{\cQ}^*,\bbeta^*)~\le~ C_{\delta}\big\{J_*/J+ \sqrt{1/({N\wedge J})}\big\},\label{eq_ill_conditioning}
\end{equation}
for some constant $C_{\delta}>0$ depending only on $\delta$.
\end{theorem}


The first part of Theorem~\ref{thm_suff_necce_concavity} provides a lower bound for $\lambda_{\min,\nu}(\Fb,\bLambda_{\cQ},\bbeta)$ within $\Xi_{\cQ}^*(c,M)$. It implies that when $J_*/J$ is lower bounded by some positive constant, one can choose a sufficiently small $c>0$ so that $\cL_{\nu}(\cdot)$ exhibits strong convexity within $\Xi_{\cQ}^*(c,M)$. The theorem’s second statement shows that this bound is tight in terms of $J_*/J$ by deriving an upper bound that matches the lower bound up to some small order terms. 


Matrix $\Sbb_m$ is used to scale the Hessian matrix $\partial^2_{\bpsi\bpsi}\cL_{\nu}(\cdot)$ so that eigenvalues of the scaled matrix are all of constant order regardless of the ratio $N/J$ as $N,J\to\infty$. Therefore, with Theorem~\ref{thm_suff_necce_concavity}, we conclude that the condition number of the scaled Hessian matrix is of constant order when $J_*$ is comparable to $J$. This property further motivates the use of a gradient descent method, whose details are presented later in Section~\ref{sec_gdalgorithm}.

There are several challenges in deriving Theorem~\ref{thm_suff_necce_concavity}. First, both $L(\Fb, \bLambda_{\cQ}, \bbeta)$ and $P(\Fb)$ are non-convex even at the true parameters. In particular, 
$\lambda_{\min}\big\{\Sbb_m\partial^2_{\bpsi\bpsi} L(\Fb^*, \bLambda_{\cQ}^*, \bbeta^*) \Sbb_m\big\} = O_p\big((N \wedge J)^{-1/2}\big)$ and $\lambda_{\min}\big\{\partial_{\bpsi\bpsi}^2P(\Fb^*)\big\} =  0$.
There is no simple complementary relation between the column space of the Hessian of $P(\Fb)$ and the ill-conditioned directions of the Hessian matrix $\Sbb_m \partial^2_{\bpsi\bpsi} L(\Fb^*, \bLambda_{\cQ}^*, \bbeta^*) \Sbb_m$ to ensure the summation is strictly positive definite. To see this, we note that the corresponding ill-conditioned directions involve a large proportion of $\bLambda_{\cQ}$, whereas $P(\Fb)$ depends solely on $\Fb$.
A more intricate challenge is that both the number of vanishing eigenvalues of $\Sbb_m\partial^2_{\bpsi\bpsi} L(\Fb^*, \bLambda_{\cQ}^*, \bbeta^*) \Sbb_m$ and their associated eigenvectors vary significantly with the block structure $\cQ$. In addition, the relationship between $\cQ$ and the penalty $P(\Fb)$, which depends on $\cM$, is implicitly confined by the $\cMQ$ Condition. To address these challenges, we consider an enlarged parameter space by freeing all the loading parameters, and construct an expanded matrix for $\partial_{\bpsi\bpsi}P(\Fb)$, denoted by $\Hb_{v}$ (see Lemma~\ref{lemma_hessian_property1} in the Supplementary Material). The matrix $\Hb_v$ preserves the information in penalty $P(\Fb)$ and its expansion component is carefully designed to capture information in the block structure. As shown in Lemma 1, this construction offsets the ill-conditioned directions of $\Sbb_m \partial^2_{\bpsi\bpsi} L(\Fb^*, \bLambda_{\cQ}^*, \bbeta^*) \Sbb_m$ within the enlarged parameter space, thereby establishing a proper lower bound on the eigenvalues. This lower bound acts as a critical intermediary step in proving Theorem~\ref{thm_suff_necce_concavity}.

\begin{remark}\it
    Theorem~\ref{thm_suff_necce_concavity} also implies that structural identifiability is necessary, though not sufficient, for local strong convexity of the function $\cL_{\nu}(\cdot)$. Specifically, when the given orthogonality constraint $\cM$ and block structure $\cQ$ do not meet the $\cMQ$ Condition, the model is not structurally identifiable by Theorem~\ref{thm_id_iff}. In this case, $J_*$ can be taken as $0$, and \eqref{eq_ill_conditioning} implies that $\lambda_{\min,\nu}(\Fb^*,\bLambda_{\cQ}^*,\bbeta^*)\overset{p}{\to}0$ as $N,J\to\infty$.
\end{remark}

\subsection{Global Optimality}\label{subsec_global_opt}
Section~\ref{subsec_local_conv} establishes the strong convexity of $\cL_{\nu}(\cdot)$ inside $\Xi_{\cQ}^*(c,M)$, ensuring that it attains a unique local minimum in this region given $J_*/J$ is non-diminishing. In this subsection, we further study this local optimum and demonstrate that it coincides with the global optimum of both \eqref{eq_constrained_maximum_likelihood_estimation} and \eqref{key_optimization_problem}.
 We first introduce the following scaling conditions.

\begin{assumption}\label{assumption_scaling}\it
For $J_*$ in \eqref{eq_spec_J_star}, we assume $J_*\asymp J$.
Furthermore, there exist sequences $\epsilon_N\to\infty$ and $\epsilon_J\to\infty$ as $N$ and $J\to\infty$ such that\begin{equation*}
        \log\big(N + J\big)\Big(\frac{\epsilon_J\log N + \epsilon_N\log J}{N\wedge J}\Big)^{1/2}\to ~0\text{, as }N,J\to\infty.
    \end{equation*}
\end{assumption}
\begin{remark}\it As implied by Theorem~\ref{thm_suff_necce_concavity}, the first part of Assumption~\ref{assumption_scaling} ensures the strong convexity of $\cL_{\nu}(\cdot)$ inside $\Xi_{\cQ}^*(c,M)$ for sufficiently small $c$. This, in turn, ensures the existence and uniqueness of the minimizer of $\cL_{\nu}(\cdot)$ in this region. The second part of Assumption~\ref{assumption_scaling} is a technical regularity condition required to derive the first-order optimality condition for this minimizer. 
It allows for a wide range of scalings in $N,J$, as $\epsilon_N$ $(\epsilon_J)$ can diverge at an arbitrarily slow rate, e.g., $\log N$ or $\log \log N$ $(\text{similarly, }\log J$ or $\log \log J)$. 
\end{remark}

The next result shows that, for any constant $c$, the minimizer of $\cL_{\nu}(\cdot)$ in $\Xi_{\cQ}^*(c,M)$ corresponds to the global optimum of $\cL_{\nu}(\cdot)$ over the feasible set in \eqref{key_optimization_problem}, and equals the constrained maximum likelihood estimator in \eqref{eq_constrained_maximum_likelihood_estimation}.
\begin{theorem}[Global Optimality]\label{thm_global}\it
    Let $(\hat\Fb,\hat\bLambda_{\cQ},\hat\bbeta)$ and $ (\tilde\Fb,\tilde\bLambda_{\cQ},\tilde\bbeta)$ be the solution of \eqref{eq_constrained_maximum_likelihood_estimation} and \eqref{key_optimization_problem}, respectively. Suppose Assumptions~\ref{assumption_psd_covariance}--\ref{assumption_scaling} and the $\cMQ$ Condition hold. Then for any $\delta>1$ and any constant $\nu>0$, when $N$ and $J$ are sufficiently large, for some small constant $c>0$, with probability at least $1-(N\wedge J)^{-\delta}$,
    \begin{equation}
        (\hat\Fb,\hat\bLambda_{\cQ},\hat\bbeta) ~=~ (\tilde\Fb,\tilde\bLambda_{\cQ},\tilde\bbeta) ~= \mathop{\arg\min}_{(\Fb,\bLambda_{\cQ},\bbeta)\in\Xi_{\cQ}^*(c,M)} \cL_{\nu}(\Fb,\bLambda_{\cQ},\bbeta).\label{eq_equi_thm}
    \end{equation}
\end{theorem}
Establishing Theorem~\ref{thm_global} consists of two steps. First, we show that the global optimum $(\tilde\Fb,\tilde\bLambda_{\cQ},\tilde\bbeta)$ lies in $\Xi_{\cQ}^*(c,M)$ for some small constant $c$, and so it consistently estimates the true parameters. Since the Lagrangian-type formulation remains non-convex over the entire feasible region, identifying the global minimizer among many potential local minima is highly non-trivial. To address this,
we take a novel approach by first constructing a reference parameter set inside $\Xi_{\cQ}^*(c,M)$ whose (negative) log-likelihood is nearly minimal over the whole feasible region of \eqref{key_optimization_problem} (see Lemma~\ref{lemma_init_reference} in the Supplementary Material). Its closeness to the true parameters forces its penalty to be small. With this reference parameter set, we can determine that the global minimizer $(\tilde\Fb,\tilde\bLambda_{\cQ},\tilde\bbeta)$ should carry a similarly small penalty and therefore nearly satisfy constraint $\cM$ and \eqref{eq_cond_normal}. Based on that, we establish  $(\tilde\Fb,\tilde\bLambda_{\cQ},\tilde\bbeta)\in\Xi_{\cQ}^*(c,M)$, which, together with Theorem~\ref{thm_suff_necce_concavity}, further implies the second equality in \eqref{eq_equi_thm}. 

The second step involves establishing the equivalence between $ (\hat\Fb,\hat\bLambda_{\cQ},\hat\bbeta)$ and $(\tilde\Fb,\tilde\bLambda_{\cQ},\tilde\bbeta)$. By our construction, this equivalence can be verified by showing that the global minimizer $(\tilde\Fb, \tilde\bLambda_{\cQ}, \tilde\bbeta)$ exactly satisfies the constraints in the original optimization problem. However, this step remains challenging, since the first step only guarantees that $P(\tilde\Fb)$ is small, but not necessarily zero. To address this issue, we develop a constructive approach showing that whenever $P(\tilde\Fb) \neq 0$, a set of parameters can be constructed within the feasible region that attains a strictly smaller value of $\cL_{\nu}(\cdot)$, contradicting the optimality of $(\tilde\Fb, \tilde\bLambda_{\cQ}, \tilde\bbeta)$. Notably, this property does not extend to the entire feasible region of \eqref{key_optimization_problem}, so this step must be carried out after establishing the result in the first step. 


\section{Estimation Consistency and Statistical Inference}\label{sec_mainres}

In this section, we establish the consistency and asymptotic distributions of the estimator $(\hat\Fb,\hat\bLambda_{\cQ},\hat\bbeta)$ obtained in \eqref{eq_constrained_maximum_likelihood_estimation}.
The following theorem states the consistency result.

\begin{theorem}[Non-asymptotic Error Bounds]\label{thm_precise_consis}
    \it  
    Suppose Assumptions~\ref{assumption_psd_covariance}--\ref{assumption_scaling} and the $\cMQ$ Condition hold. Let $(\hat\Fb,\hat\bLambda_{\cQ},\hat\bbeta)$ be the estimator defined in \eqref{eq_constrained_maximum_likelihood_estimation}. For any $\delta>1$ and some constant $C_{\delta}$ depending only on $\delta$, when $N$ and $J$ are large enough, with probability at least $1 - (N\wedge J)^{-\delta}$, the following bounds hold:
    \begin{equation*}
       {J}^{-1}\|\hat\bLambda_{\cQ} - \bLambda_{\cQ}^*\|_F^2 +  {J}^{-1}\|\hat\bbeta - \bbeta^*\|^2\le C_{\delta}(N^{-1}+J^{-2}),\; {N}^{-1}\|\hat\Fb - \Fb^*\|_F^2\le C_{\delta}(N^{-2}+J^{-1});\end{equation*} and \begin{equation*}
       \|\hat\bLambda_{\cQ} - \bLambda_{\cQ}^*\|_{\infty} + \|\hat\bbeta-\bbeta^*\|_{\infty} ~\le~ C_{\delta}\frac{\epsilon_{NJ}\log J}{(N\wedge J)^{1/2}},\; \|\hat\Fb-\Fb^*\|_{\infty}~\le~ C_{\delta} \frac{\epsilon_{NJ}\log N}{(N\wedge J)^{1/2}},
    \end{equation*}
    where 
    $\epsilon_{NJ}= (\epsilon_N\log J+\epsilon_J\log N)^{1/2}$ with $\epsilon_N$ and $\epsilon_J$ specified as in Assumption~\ref{assumption_scaling}.
\end{theorem}

Theorem~\ref{thm_precise_consis} establishes non-asymptotic $\ell_2$-error bounds that improve upon the previously obtained $O_p(N^{-1}+J^{-1})$ rate~\citep[e.g.,][]{chen2020structured,Wang2018Maximum}.
It also provides the $\ell_{\infty}$-error bound that applies to the loading parameters across all blocks, regardless of their sizes. We emphasize that the strong convexity established in Theorem~\ref{thm_suff_necce_concavity} alone is insufficient to achieve these results. A key ingredient in establishing Theorem~\ref{thm_precise_consis} is a precise characterization of the inverse of the high-dimensional Hessian matrix $\partial^2_{\bpsi\bpsi} \cL_{\nu}(\cdot)$. Since the Hessian matrix exhibits different properties for the latent factor and loading parameter components, we partition it into corresponding blocks and establish sharp $\ell_2$- and $\ell_{\infty}$-error bounds for each matrix block (see Lemma~\ref{lemma_hessian_property12} in the Supplementary Material). These blockwise results are then used to derive an explicit expression for the inverse Hessian and to obtain the sharp error bounds. 



Across a broad range of asymptotic regimes, the results in Theorem~\ref{thm_precise_consis} attain the oracle $\ell_2$ rates for the latent factors and loading parameters when the other component is known. Specifically, when $\Fb^*$ is given, the estimation of $\hat\bLambda_{\cQ}^*$ and $\hat\bbeta^*$ reduces to a generalized linear model setting, under which the oracle rate is $O_p(N^{-1})$. When $N = O(J^2)$, the consistency rate for the loading parameters and intercepts achieves this oracle rate. Similarly, for the latent factors, when $J = O(N^2)$, their consistency rate achieves the oracle rate of $O_p(J^{-1})$. 



Next, we present the limiting distributions of the estimators $(\hat\Fb,\hat\bLambda_{\cQ},\hat\bbeta)$. For each $k\in[K]$ and $j\in\cJ_k$, define the vector of free loading parameters as $\blambda_{j(k)}:= (\lambda_{jd})_{d\in\cA_k}$. Denote the estimated vector by $\hat\blambda_{j(k)} = (\hat\lambda_{jd})_{d\in\cA_k} ^\T $ with $\hat\lambda_{jd} = (\hat\bLambda_{\cQ})_{[j,d]}$. 
The subvector of latent factors associated with response to item $j$ is denoted by $\bbf_{i(k)} = (\bbf_i)_{\cA_k}$, with the corresponding true value denoted by $\bbf_{i(k)}^*= (\bbf_i^*)_{\cA_k}$. Let $A_k = |\cA_k|$ be the number of latent factors associated with block $k$. 

\begin{theorem}[Asymptotic Distributions]\label{thm_asymptotic_normality}\it 
    Suppose Assumptions~\ref{assumption_psd_covariance}--\ref{assumption_scaling} and the $\cMQ$ Condition hold. For $(\hat\Fb,\hat\bLambda_{\cQ},\hat\bbeta)$ obtained in \eqref{eq_constrained_maximum_likelihood_estimation}, as $N,J\to\infty$, if $N=o(J^2)$, it holds that \begin{equation*}
        \big(\bPhi_{j(k)}^*\big)^{-1/2}\begingroup
    \renewcommand{\arraystretch}{0.5}\begin{pmatrix}
            \hat\beta_j - \beta_j^*\\\hat\blambda_{j(k)} - \blambda_{j(k)}^*
        \end{pmatrix}\endgroup ~ \overset{d}{\to} ~\cN(\zero_{A_k+1},\Ib_{A_k+1}\big)\text{ for each }k\in[K]\text{ and }j\in J_k,
    \end{equation*}
    where letting $\bbf_{i(k)}^{c*} = (1, (\bbf_{i(k)}^*)^\T)^\T$, $l_{ij(k)}^{\prime *} = l_{ij}^{\prime}(w_{ij(k)}^*)$, and $l_{ij(k)}^{\prime\prime *} = l_{ij}^{\prime\prime}(w_{ij(k)}^*)$, $\bPhi_{j(k)}^*$ is given by
    \begin{equation}
    \Big[\sum_{i\in[N]}-l_{ij(k)}^{\prime\prime*}\bbf_{i(k)}^{c*}\big\{\bbf_{i(k)}^{c*}\big\}^\T\Big]^{-1}\Big[\sum_{i\in[N]}(l_{ij(k)}^{\prime*})^2\bbf_{i(k)}^{c*}\big\{\bbf_{i(k)}^{c*}\big\}^\T\Big]\Big[\sum_{i\in[N]}-l_{ij(k)}^{\prime\prime*}\bbf_{i(k)}^{c*}\big\{\bbf_{i(k)}^{c*}\big\}^\T\Big]^{-1} \label{eq_sandwich}
\end{equation} 
For the estimates of latent factors, as $N,J\to\infty$, if $J=o(N^2)$,  it holds that
    \begin{equation*}
        \big(\bPsi_{i}^*\big)^{-1/2}\big(\hat\bbf_i - \bbf_i^*\big)~\overset{d}{\to}~\cN(\zero_{D},\Ib_{D}\big),\text{ for each }i\in[N],
    \end{equation*}
    where $\bPsi_i^*$ is given as 
\[
    \Big\{\!\!\sum_{k\in[K]}\!\sum_{j\in\cJ_k}-l_{ij(k)}^{\prime\prime*}\blambda_{j}^{c*}(\blambda_{j}^{c*})^\T\!\Big\}^{-1}\Big\{\!\!\sum_{k\in[K]}\!\sum_{j\in\cJ_k}(l_{ij(k)}^{\prime*})^2\blambda_{j}^{c*}(\blambda_{j}^{c*})^\T\!\Big\}\Big\{\!\!\sum_{k\in[K]}\!\sum_{j\in\cJ_k}-l_{ij(k)}^{\prime\prime*}\blambda_{j}^{c*}(\blambda_{j}^{c*})^\T\!\Big\}^{-1},
\]
with $\blambda_{j}^{c*} = (\bLambda_\cQ^*)_{[j,]}^\T$ and $k$ indexing the block that contains $j$.
Furthermore, $\bPhi_{j(k)}^*$ and $\bPsi_i^*$ can be consistently estimated by the plug-in estimators $\hat\bPhi_{j(k)}$ and $\hat\bPsi_i$ that replace $(\Fb^*,\bLambda_{\cQ}^*,\bbeta^*)$ with the estimation $(\hat\Fb,\hat\bLambda_{\cQ},\hat\bbeta)$, and the above asymptotic normality results hold with $\bPhi_{j(k)}^*$ and $\bPsi_i^*$ replaced by $\hat\bPhi_{j(k)}$ and $\hat\bPsi_i$, respectively.
\end{theorem}
The asymptotic covariance matrices in Theorem~\ref{thm_asymptotic_normality} attain the oracle Cramer-Rao information lower bound. Specifically, for each item $j$, define the parameter vector $\bxi_{j(k)}:=(\beta_j,\blambda_{j(k)}^\T)^\T$. When all latent factors $\Fb^*$ are assumed to be known, for each $j \in \cJ_k$ with $k \in [K]$, the Fisher information matrix for $\bxi_{j(k)}$ is given by
\begin{equation*}
    \Ib(\bxi_{j(k)}) := \EE\Big[\frac{\partial L(\Fb^*,\bLambda_{\cQ},\bbeta)}{\partial\bxi_{j(k)}}\frac{\partial L(\Fb^*,\bLambda_{\cQ},\bbeta)}{\partial\bxi_{j(k)}}^\T\Big] = \EE\Big[\sum_{i=1}^N\big[l_{ij}^{\prime}(\bxi_{j(k)}^\T\bbf_{i(k)}^{c*})\big]^2\bbf_{i(k)}^{c*}(\bbf_{i(k)}^{c*})^\T\Big].
\end{equation*}
Here, the expectation is taken over the responses, and the second equality follows from their conditional independence. The inverse of the Fisher information at the true parameters, $\{\Ib(\bxi_{j(k)}^*)\}^{-1}$, is asymptotically equivalent to $\bPhi^*_{j(k)}$. Consequently, when $N=o(J^2)$, the estimators for each item-specific parameter vector $\bxi_{j(k)}$ are asymptotically efficient. Notably, this asymptotic efficiency holds across all items, regardless of the block sizes to which they belong. 
Analogously, when $\bLambda_{\cQ}^*$ and $\bbeta^*$ are assumed to be known, $\bPsi^*_{i}$ is asymptotically equivalent to the inverse Fisher information for $\bbf_i^*$. Thus, when $J=o(N^2)$, the estimators for the latent factors are likewise asymptotically efficient. The result is also of practical importance. Because the block structure and imposed identifiability constraints relate the latent dimensions to scientifically meaningful constructs, the individual loading parameters and latent factors often carry substantive meaning. Our results therefore provide practically meaningful and asymptotically efficient inference for these quantities.


\begin{remark}\it
In Theorem~\ref{thm_asymptotic_normality}, we adopt scaling conditions $N = o(J^2)$ for asymptotic distributions of loading parameters and intercepts and $J = o(N^2)$ for those of factors. These conditions are mild and accommodate a variety of asymptotic sequences of $N$ and $J$. They ensure certain small-order terms are negligible as $N,J\to\infty$, thereby ensuring the validity of the asymptotic distributions. Similar assumptions are commonly considered in the literature~\citep{bai2003inferential,bai2012statistical,Wang2018Maximum,cui2025identifiability}. \end{remark}


\section{Algorithm and Guarantees}\label{sec_gdalgorithm}

We now introduce a first-order algorithm to obtain the estimator $(\hat\Fb,\hat\bLambda_{\cQ},\hat\bbeta)$ in \eqref{eq_constrained_maximum_likelihood_estimation}. Building on the equivalence established in Theorem~\ref{thm_global}, $(\hat\Fb,\hat\bLambda_{\cQ},\hat\bbeta)$ can be obtained via minimizing $\cL_{\nu}(\cdot)$ within $\Xi_{\cQ}^*(c,M)$. The strong convexity established in Theorem~\ref{thm_suff_necce_concavity} motivates the use of a first-order method, ensuring both computational efficiency and provable convergence. The detailed description is given in Algorithm~\ref{alg:non-convex}.

The initialization required by Algorithm~\ref{alg:non-convex} is described in Section~\ref{supp_sec_gdalgorithm} of the Supplementary Material, where we also show that the resulting initialization lies within $\Xi_{\cQ}^*(c,M)$ with high probability.
 Starting from the initial estimates, the algorithm iteratively updates all estimates along the gradient direction $\partial_{\bpsi}\cL_{\nu}(\cdot)$ with step size
\begin{equation*}
    \eta_{f} = \eta/J,\quad\eta_{\lambda} = \eta/N,\quad\eta_{\beta} = \eta/N,
\end{equation*} for some small pre-specified constant $\eta>0$. The rescaling accounts for different parameter dimensions of $\Fb$, $\bLambda_{\cQ}$, and $\bbeta$, which may diverge at different rates. The gradients for loading parameters and intercepts are computed in a block-wise manner. For each $k\in[K]$, the algorithm updates the estimates of the loading and intercept parameters in block $k$ in Steps~7 and~8, respectively. Since the latent factors $\Fb$ are shared across many blocks, their update in Step~10 is aggregated over blocks. The gradient with respect to $\Fb$ is computed as
 \begin{align*}
     \partial_{\Fb}\cL_{\nu}\big(\hat\Fb^{(t)},\hat\bLambda_{\cQ}^{(t)}, \hat\bbeta^{(t)}\big) ~=~\sum_{k=1}^K\tilde\Lb_{k}^{(t)}\hat\bLambda_{k}^{(t)}(\Ib_D)_{[\cA_k,]} ~+~ \nu\partial_{\Fb}P(\hat\Fb^{(t)}).
 \end{align*}\begin{algorithm}
   \caption{ Block-wise Gradient Descent Scheme.} 
   \label{alg:non-convex}
\begin{algorithmic}[1]
\STATE {\bfseries Input:} initial estimates: $(\hat\Fb^{(0)},\hat\bLambda^{(0)}_{\cQ},\hat\bbeta^{(0)})$; step sizes: $\eta_f, \eta_\lambda,  \eta_\beta$; block structure: $\cQ$; orthogonality constraint: $\cM$; maximum number of iterations $T$.\\{\bfseries Output:} $(\hat\Fb^{(T)},\hat\bLambda_{\cQ}^{(T)},\hat\bbeta^{(T)})$.
\STATE for each $k\in[K]$, set $\hat\bLambda_k^{(0)} = \big(\bLambda_{\cQ}^{(0)}\big)_{[\cJ_k,\cA_k]}$ and $\hat\bbeta^{(0)}_k = \bbeta^{(0)}_{\cJ_k}$;
  \FOR{$t=0, 1, \dots, T-1$ }
  \FOR{$k = 1,2,\dots,K$}
  \STATE $\tilde\bTheta_k^{(t)} = \hat\Fb^{(t)}_{[,\cA_k]}\big\{\hat\bLambda_{k}^{(t)}\big\}^\T + \one_N\big\{\hat\bbeta^{(t)}_k\big\}^\T$;
  \STATE $\tilde \Lb_{k}^{(t)} = -\left\{l_{ij}^{\prime}\big(\{\tilde \bTheta^{(t)}_k\}_{[i,j]}\big)\right\}_{i\in[N],j\in\cJ_k}$;
  \STATE $\tilde\bLambda_{k}^{(t+1)} = \hat\bLambda_{k}^{(t)} -\eta_{\lambda} \big\{\tilde \Lb_{k}^{(t)}\big\}^\T\hat\Fb^{(t)}_{[,\cA_k]}$;
  \STATE $\tilde\bbeta^{(t+1)}_k = \hat\bbeta^{(t)}_k - \eta_{\beta} \big\{\tilde \Lb_{k}^{(t)}\big\}^\T\one_N$;
  \ENDFOR
  \STATE $\tilde\Fb^{(t+1)} = \hat\Fb^{(t)} -\eta_{f}\sum_{k=1}^{K}\tilde \Lb_{k}^{(t)}\hat\bLambda_{k}^{(t)}(\Ib_D)_{[\cA_k,]} -\eta_{f}\nu\partial_{\,\Fb} P(\hat\Fb^{(t)})$;
  \STATE truncate the parameters to $[-M,M]$: for each $k\in[K]$, $\hat\bLambda_{k}^{(t+1)} = \cT_M\big(\tilde\bLambda_{k}^{(t + 1)}\big)$, $\hat\bbeta_k^{(t+1)} = \cT_M\big(\tilde\bbeta_k^{(t + 1)}\big)$, and $ \hat\Fb^{(t+1)} = \cT_M\big(\widetilde{\Fb}^{(t+1)}\big)$, with $\cT_M(\cdot)$ the corresponding truncation function; \ENDFOR
  \STATE assemble $\hat\bLambda_{\cQ}^{(T)} \leftarrow \big\{\hat\bLambda_{k}^{(T)}\big\}_{k\in[K]}$ and $\hat\bbeta^{(T)} \leftarrow \big\{\hat\bbeta^{(T)}_k\big\}_{k\in[K]}$.
\end{algorithmic}
\end{algorithm}
 Here, $\partial_{\bF}P(\bF)$ denotes an $N\times D$ matrix $(\partial_{f_{id}}P(\Fb))_{N\times D}$, and $(\Ib_D)_{[\cA_k,]}$ maps to the corresponding factors associated with block $k$. Theorem~\ref{thm_suff_necce_concavity} provides a range of $\nu$ for which the Lagrangian-type objective is strongly convex. We set $\nu=0.5$ throughout implementation, and the numerical results are insensitive to this choice. To satisfy the inequality constraints in \eqref{key_optimization_problem}, a truncation operator $\cT_M(\cdot)$ is then applied using the same constant $M$  as in \eqref{key_optimization_problem}. This operator truncates each parameter value at $M$ if it exceeds $M$, and at $-M$ if it falls below $-M$. After the $T$-th iteration, we assemble $\{\hat\bLambda_k^{(T)}\}_{k\in[K]}$ and $\{\hat\bbeta_k^{(T)}\}_{k\in[K]}$ into $\hat\bLambda_{\cQ}^{(T)}$ and $\hat\bbeta^{(T)}$, where for each $k\in[K]$, $(\bLambda_{\cQ}^{(T)})_{[\cJ_k,\cA_k]} = \hat\bLambda_k^{(T)}$ and $\bbeta^{(T)}_{\cJ_k} = \hat\bbeta^{(T)}_k$. The final estimate is then given by $(\hat\Fb^{(T)}, \hat\bLambda_{\cQ}^{(T)}, \hat\bbeta^{(T)})$. At each iteration $t$, we similarly define the assembled iterates for the loading parameters and intercepts by $\hat\bLambda_{\cQ}^{(t)}$ and $\hat\bbeta^{(t)}$, respectively.


Now we establish the error bounds for the iterates in Algorithm~\ref{alg:non-convex}. Define the $\ell_2$- and $\ell_{\infty}$-error between the optimizer $(\hat\Fb, \hat\bLambda_{\cQ}, \hat\bbeta)$ in \eqref{eq_constrained_maximum_likelihood_estimation} and the estimates at the $t$-th iteration as
        \begin{equation*}
             \hat e_{\ell,2}^{(t)} ~=~ {N}^{-1}\|\hat\Fb^{(t)} - \hat\Fb\|_F ^2 + {J}^{-1}\|\hat\bLambda_{\cQ}^{(t)} - \hat\bLambda_{\cQ}\|_F^2 + {J}^{-1}\|\hat\bbeta^{(t)} - \hat\bbeta\|_F^2,
        \end{equation*}
        \begin{equation*}
             \hat e_{\ell,\infty}^{(t)} ~=~ \|\hat\Fb^{(t)} - \hat\Fb\|_{\infty} + \|\hat\bLambda_{\cQ}^{(t)} - \hat\bLambda_{\cQ}\|_{\infty}  + \|\hat\bbeta^{(t)} - \hat\bbeta\|_{\infty},
        \end{equation*}
        respectively.
        The following theorem demonstrates the linear convergence of Algorithm~\ref{alg:non-convex}.
\begin{proposition}[Linear Convergence]\label{thm_linear_convergen}
    Suppose Assumptions~\ref{assumption_psd_covariance}--\ref{assumption_scaling} and the $\cMQ$ Condition hold. Let the initial estimates $(\hat\Fb^{(0)},\hat\bLambda_{\cQ}^{(0)}, \hat\bbeta^{(0)})$ be obtained from Algorithm~\ref{alg:init} below. Fix any $\delta>1$ and choose $\eta$ such that $\eta J_*/J \le 1/(2\gamma_\delta)$, where $\gamma_\delta$ is the constant from Theorem~\ref{thm_suff_necce_concavity}. Then, with probability at least $1 - (N\wedge J)^{-\delta} - C\exp(-c_0\epsilon_N) - C\exp(-c_0\epsilon_J)$, the iterates produced by Algorithm~\ref{alg:non-convex} converge linearly to $(\hat\Fb, \hat\bLambda_{\cQ}, \hat\bbeta)$: for all $t\ge 0$,
    \begin{equation}
             \hat e_{\ell, 2}^{(t)}~\le~ C_0\big(1-\eta C_1J_*/J\big)^{2t}\kappa_{NJ}^2 ,\;\text{ and }~\;\hat e_{\ell, \infty}^{(t)}~\le~ C_0\big(1-\eta C_1J_*/J\big)^{t}\log(N\vee J)\kappa_{NJ},\label{eq_l_2_cont_star}\end{equation}
    for a constant $C_1>0$ depending only on $\delta$, where $\kappa_{NJ} = \log^2(N\vee J)/\sqrt{N\wedge J}$.
\end{proposition}
\begin{proof}
    See Section~\ref{supp_sec_prove_thm_linear_convergen}.
\end{proof}
Proposition~\ref{thm_linear_convergen} establishes the linear convergence of Algorithm~\ref{alg:non-convex} in both $\ell_2$- and $\ell_{\infty}$-norms under a constant step size $\eta$. The constant $C_1$ in the proposition depends on the curvature parameter $\gamma_\delta$ from Theorem~\ref{thm_suff_necce_concavity}, which is assumed to be constant by Assumption~\ref{assumption_scaling}. The estimation errors between $(\hat\Fb^{(t)}, \hat\bLambda_{\cQ}^{(t)},\hat\bbeta^{(t)})$ and the true parameters follow directly from Theorem~\ref{thm_precise_consis} and Proposition~\ref{thm_linear_convergen}: with probability at least $1-2(N\wedge J)^{-\delta}$, for some constant $C_{\delta}$ depending on $\delta$, we have
\begin{equation*}
    {J}^{-1}\|\hat\bLambda_{\cQ}^{(t)} - \bLambda_{\cQ}^*\|_F^2 +  {J}^{-1}\|\hat\bbeta^{(t)} - \bbeta^*\|^2 ~\le~ C_0\big(1-\eta C_1J_*/J\big)^{2t}\kappa_{NJ}^2 + C_{\delta}\big(N^{-1} + J^{-2}\big);
\end{equation*}
\begin{equation*}
    \|\hat\bLambda_{\cQ}^{(t)} - \bLambda_{\cQ}^*\|_{\infty} + \|\hat\bbeta^{(t)}-\bbeta^*\|_{\infty} ~\le~ C_0\big(1-\eta C_1J_*/J\big)^{t}\log(N\vee J)\kappa_{NJ} + C_{\delta}\frac{\epsilon_{NJ}\log J}{(N\wedge J)^{1/2}}.
\end{equation*}
Since $T$ can be arbitrarily large, $ \hat\bLambda_{\cQ}^{(t)}$ and $\hat\bbeta^{(t)}$ achieve the same statistical optimality as the original estimator in \eqref{eq_constrained_maximum_likelihood_estimation}. Similar optimality holds for $\hat\Fb^{(t)}$.
In Proposition~\ref{thm_linear_convergen}, the contraction rates of both $\hat e_{\ell,2}^{(t)}\to 0$ and $\hat e_{\ell,\infty}^{(t)}\to 0$ depend on the ratio $J_*/J$, reflecting the convexity results established in Theorem~\ref{thm_suff_necce_concavity}. The term $\kappa_{NJ}$ captures the statistical error inherited from the initialization (see Proposition~\ref{thm_consistent_alg_init} below). This rate does not affect either the statistical optimality or the convergence speed of the algorithm, as long as the initialization yields an estimate within $\Xi_{\cQ}^*(c,M)$.

With Theorem~\ref{thm_asymptotic_normality} and Proposition~\ref{thm_linear_convergen}, we immediately have the following result.
\begin{corollary}\label{coro_linear_convergence}
Under the setup in Proposition~\ref{thm_linear_convergen}, when the number of iterations $T$ satisfies
\begin{equation*}
    \frac{1}{T}\cdot\frac{\log\big\{C_0\log(N\vee J)\kappa_{NJ}(N\vee J)^{1/2}\big\}}{|\log (1-\eta C_1J_*/J)|}~ \to ~0~\text{  as }~N,J\to\infty,
\end{equation*}for constants $C_0$ and $C_1$ specified in Proposition~\ref{thm_linear_convergen},
the asymptotic results in Theorem~\ref{thm_asymptotic_normality} apply for the estimates $(\hat\Fb^{(T)},\hat\bLambda_{\cQ}^{(T)},\hat\bbeta^{(T)})$ obtained from Algorithm~\ref{alg:non-convex}. 
\end{corollary}

After ignoring the $\log\log$-terms, and with the scaling conditions $N=o(J^2)$ and $J=o(N^2)$ required by Theorem~\ref{thm_asymptotic_normality}, the asymptotic results in Theorem~\ref{thm_asymptotic_normality}  hold for the $T$-th step estimators  $(\hat\Fb^{(T)},\hat\bLambda_{\cQ}^{(T)},\hat\bbeta^{(T)})$ as long as $T\gg \log (N\vee J)$. 

There have been extensive studies on statistical properties of gradient descent sequences for similar non-convex problems involving matrix factorization~\citep{ma2018implicit,ma2020universal,chen2021bridging,yan2024inference}. However, most existing studies focus on the linear setting, where the analysis primarily relies on a fine-grained decomposition of the estimation error $\hat\bLambda_{\cQ}^{(t)} - \bLambda_{\cQ}^*$. In nonlinear settings, however, such techniques are inapplicable because the first-order updates no longer yield a linear form for $\hat\bLambda_{\cQ}^{(t)} - \bLambda_{\cQ}^*$, making the noise component difficult to disentangle. Instead, our analysis uses the Lagrangian-type formulation that benefits from the favorable properties established in Section~\ref{sec_estimation}. Remarkably, unlike algorithms that impose regularization and thereby introduce additional bias into the solution~\citep{keshavan2010matrix,sun2016guaranteed,chen2021bridging}, our Lagrangian-type formulation has exactly the same solution as the original problem. While tailored to the block structured latent variable model studied here, this technique has the potential to extend to more general and challenging settings.




\section{Simulation Study}\label{sec:simu}
In this section, we present simulation studies to validate our theoretical findings. We consider logistic latent factor models with various block structures, where the link function is $g_{ij}(y|w) = \exp(w)^y/(1+\exp(w))$ for all $i\in[N]$ and $j\in[J]$. 

We begin with introducing the four block structure designs used in the simulation study. These designs are adapted from the four setups in Figure~\ref{fig:illus}. In each simulation setup, the block structure $\cQ$ is proportional to that illustrated in the corresponding graphical representation in Figure~\ref{fig:illus}. Specifically, if $x$\% of the items belong to a block in the figure, the corresponding block in the simulation contains $x$\% of the $J$ total items, and the items’ responses are influenced by the same subset of factors. The number of latent factors and their orthogonality constraints are also preserved, as indicated in the figure. For reference, we denote the block structures by (a)–(d), following their labels in the figure.

For the data-generating process of the true parameters, the latent factors $\{\bbf_i\}_{i\in[N]}$ are first drawn i.i.d. from $\cN(\zero_D,\bSigma_{\tau})$, where $\bSigma_{\tau}\in\RR^{D\times D}$ has $(l,h)$-th entry $\tau^{|l-h|}$, with $\tau$ specified later. Then they are transformed to satisfy the constraints in \eqref{eq_cond_normal} and the orthogonal constraint $\cM$, with $\cM$ indicated for each setup in Figure~\ref{fig:illus}. 
For loading parameters, the entries in $\bLambda_k$ are independently generated from $Unif(1,2)$ for each $k\in[K]$.  The intercepts are drawn i.i.d. from $\cN(0,1)$. 

Under all four block structures, we generate the true parameters under the following settings: (1) sample size $N\in\{500, 1000, 1500, 2000\}$; (2) number of items $J \in \{500, 1000, 1500,$ $2000\}$; (3) factor correlation $\tau \in \{0.3, 0.7\}$. The estimator in \eqref{eq_constrained_maximum_likelihood_estimation} is obtained using the algorithm described in Section~\ref{sec_gdalgorithm}. To assess our asymptotic results, we compute 95\% Wald intervals for each latent factor, loading parameter, and intercept based on Theorem~\ref{thm_asymptotic_normality}. We perform 200 replications and report the empirical coverage rates. Figures~\ref{fig_factors_low} and~\ref{fig_loadings_low} present the coverage results for the latent factors and loading parameters under the low-correlation setting ($\tau = 0.3$), with panels (a)–(d) corresponding to block structures (a)–(d). Additional results for intercepts and for the high-correlation setting are provided in Section~\ref{sec_supp_simu1}
of the Supplementary Material.


{\spacingset{1.19}\begin{figure}[h]
    \centering
    \includegraphics[width=3.6in, height=0.16in]{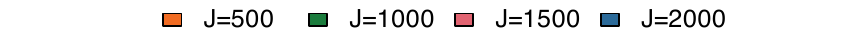}\\
    \includegraphics[width=5.6in, height=2in]{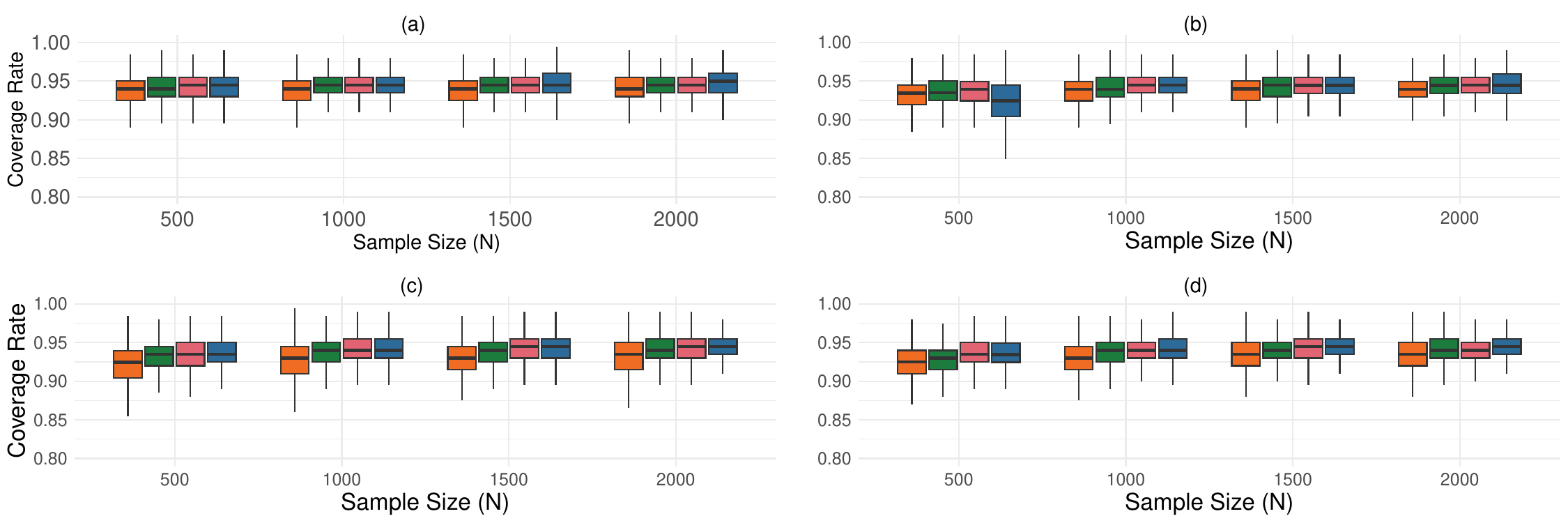}
    \caption{Empirical coverage rates of factor parameters in the low-correlation case ($\tau = 0.3$) under different $N$, $J$. Panels (a)–(d) report results for block structures (a)–(d).} 
    \label{fig_factors_low}
\end{figure}
\begin{figure}[h]
    \centering
    \includegraphics[width=3.6in, height=0.16in]{legende.pdf}
    \includegraphics[width=5.6in, height=2in]{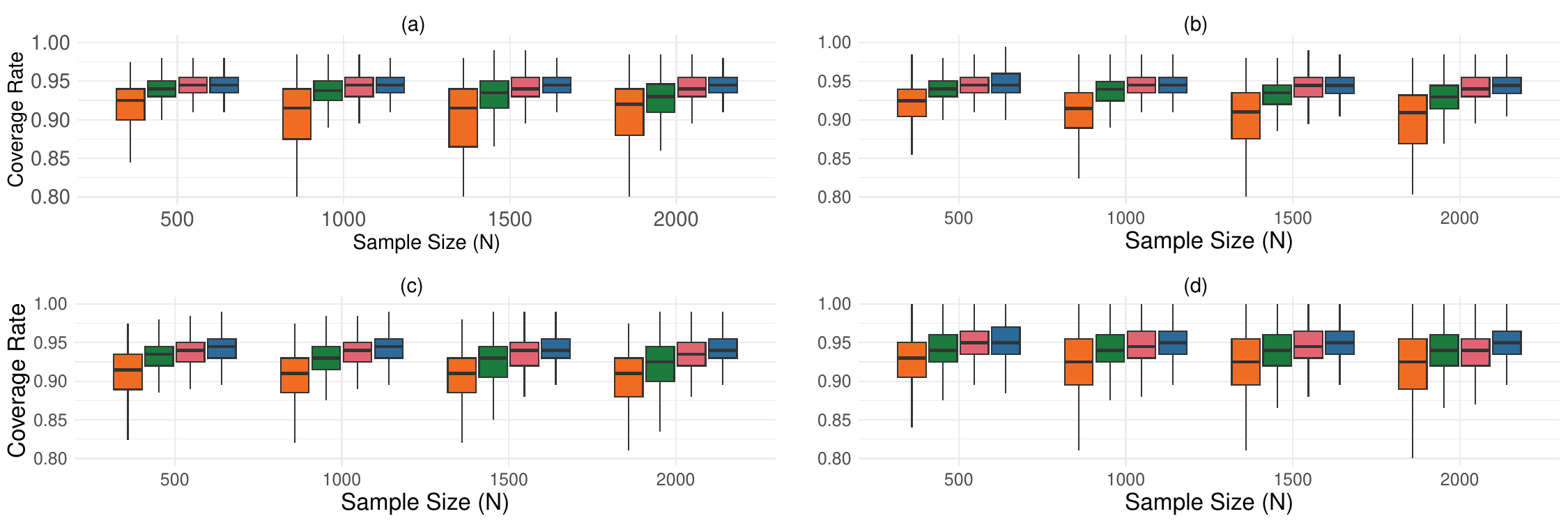}
    \caption{Empirical coverage rates of loading parameters in the low-correlation case ($\tau = 0.3$) under different $N$ and $J$. Panels (a)–(d) report results for block structures (a)–(d).}
    \label{fig_loadings_low}
\end{figure}}

Overall, the asymptotic results show good empirical coverage across all scenarios and parameters. For the estimated factor parameters in the low-correlation setting (Figure~\ref{fig_factors_low}), the coverage rates reach about 95\% in all scenarios once both $N$ and $J$ exceed 1000. Under the block structure (b), the coverage for factors is slightly below 95\% when $N=500$ and $J=2000$. Theoretically, this occurs because some higher-order terms in the asymptotic expansion of $(\bPsi_i^*)^{-1/2}(\hat\bbf_i - \bbf_i^*)$ become large when $\sqrt{J}/N$ is not sufficiently small. For the estimated loading parameters in the low-correlation case (Figure~\ref{fig_loadings_low}), the coverage rates follow a similar trend, increasing noticeably as $J$ grows from 500 to 1000 and staying around 95\% once $J$ exceeds 1000. There are only minor discrepancies in the empirical coverage rates across different settings, implying that the accuracy of the asymptotic approximations in Theorem~\ref{thm_asymptotic_normality} depends primarily on $N$ and $J$ rather than the block structure.
For the intercepts in the low‑correlation setting (Figure~S1), 
the empirical coverage rates stay close to 95\%  across all settings of $N$, $J$, and block structures. The high‑correlation setting exhibits similar coverage results across these settings (see details in Figures~S2--S4). 
Taken together, these results provide strong empirical support for our theoretical findings.

\section{Real Data Analysis}\label{sec:data}

We apply our proposed method to two datasets: the PISA 2022 educational assessment data analyzed below and a psychological assessment dataset analyzed in Section~\ref{sec_supp_data} of the Supplementary Material. PISA is a large-scale international assessment of 15-year-old students’ proficiency in mathematics, science, and reading~\citep{organisation2023pisa}. We analyze the data collected in the United States. Under PISA’s rotated-booklet design, each student is administered a prespecified subset of the full item pool. We fit the model using the observed-response likelihood and, following standard treatments in educational assessment, assume that booklet assignment is conditionally independent of the responses given the students' latent factors~\citep{chen2023statistical,ouyang2024statistical}. The resulting missingness allows our theoretical results to be extended to the maximum likelihood estimation based on the observed data likelihood. After excluding students and items with no valid responses, the resulting sample contains $N=4020$ students and $J=453$ binary items, of which items $1$--$169$, $170$--$272$, and $273$--$453$ assess mathematics, science, and reading, respectively.

Following the literature~\citep{brunner2008no,pokropek2022much}, we fit a four-dimensional bifactor model, a special case of our block structured latent variable model framework. For each student $i$, let $\bbf_i = (f_{i1},f_{i2},f_{i3},f_{i4})^\T$, where $f_{i1}$ represents a general factor capturing student $i$'s overall learning ability, and $f_{i2},f_{i3},f_{i4}$ are domain specific factors capturing mathematics, science, and reading learning abilities, respectively. Responses to all items depend on the general factor $f_{i1}$, while responses to mathematics, science, and reading items depend on $f_{i2}$, $f_{i3}$, and $f_{i4}$, respectively. For identifiability, the three domain-specific factors are constrained to be orthogonal to the general factor, consistent with the $\cMQ$ Condition.

Since the responses are binary, we fit the model under the block structure and orthogonality constraint described above using a logistic link, corresponding to the multidimensional two‑parameter logistic (M2PL) model~\citep{reckase2009}. 
The estimation is obtained using the algorithm in Section~\ref{sec_gdalgorithm}. 
We construct asymptotic confidence intervals for the loading parameters using the asymptotic distributions established in Theorem~\ref{thm_asymptotic_normality}.

{\spacingset{1.19}
\begin{figure}[h]
    \centering
    \includegraphics[width=5.3in, height=1.6in]{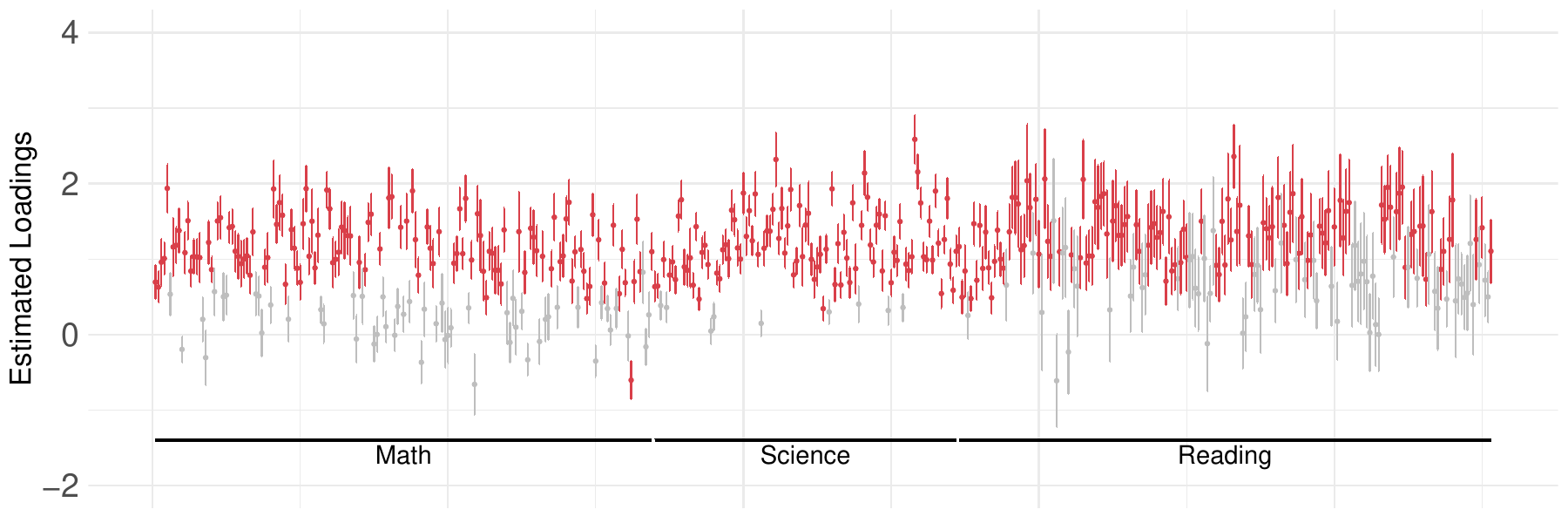}
    \caption{Confidence intervals for the loading parameters on the general factor. Intervals in red represent significant loading parameters after adjustment with Bonferroni correction.}
    \label{fig_main_laoding}
\end{figure}
\begin{figure}[h]
    \centering
    \includegraphics[width=5.3in, height=1.25in]{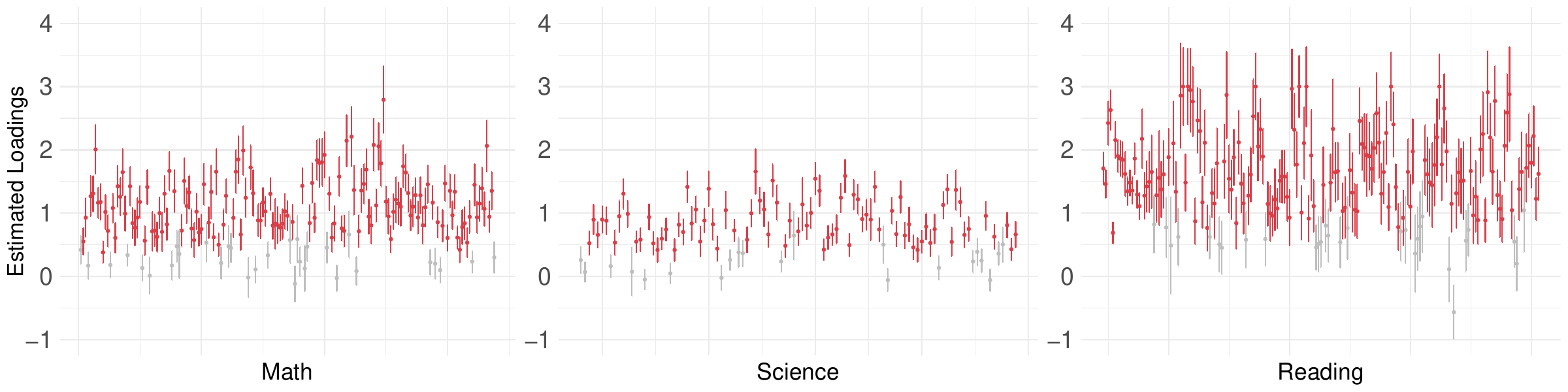}
\caption{Confidence intervals for the loading parameters on the three domain-specific factors. Intervals in red represent significant loading parameters after adjustment with Bonferroni correction.}
    \label{fig_indi_laoding}
\end{figure}}

Figure~\ref{fig_main_laoding} presents the estimated loading parameters for the general factor, and Figure~\ref{fig_indi_laoding} displays the estimated domain-specific loading parameters for mathematics, science, and reading in three separate subfigures. Significant loading parameters after Bonferroni correction at the 0.05 level are highlighted in red. The results for the latent abilities are provided in Figure~\ref{fig_pisa_factors}.

One notable pattern in Figures~\ref{fig_main_laoding} and~\ref{fig_indi_laoding} is that science items exhibit relatively larger values of loading parameters on the general factor, compared to their science-specific loading parameters. This suggests that students’ learning ability in science is largely captured by the general factor. In contrast, reading items display smaller values of general-factor loading parameters but much larger values of reading-specific loading parameters, indicating that reading skill relies more on domain-specific learning abilities and is less captured by the general factor. A similar pattern is observed by~\cite{pokropek2022much}.


Moreover, we find that all significant loadings are positive except for one mathematics item on the general factor. The positive loadings align with the model’s interpretation, as they indicate that correctly answering a testing item is associated with higher values of the latent factors representing students’ learning abilities. Notably, the significance in general-factor loadings suggests that the general factor captures broad learning skills across mathematics, science, and reading.
The only exception is the mathematics item ``Sleep and Reaction Time'', which displays a negative loading value on the general factor ($p=2.75\times 10^{-3})$. This item’s estimated loading parameter on the math‐specific factor is $0.93$ and is significant ($p = 4.94\times 10^{-10}$), indicating that it primarily assesses mathematical ability but is negatively related to the general latent factor.

{\spacingset{1.19}
\begin{figure}[h]
    \centering
    \includegraphics[width=3.6in, height=0.26in]{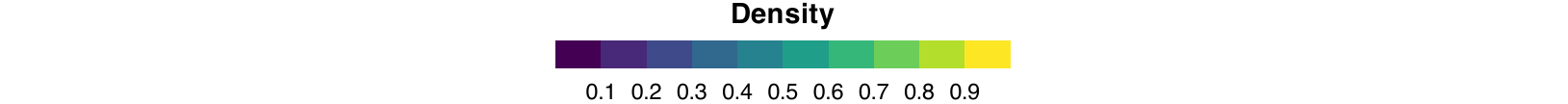}
    \includegraphics[width=5.6in, height=1.9in]{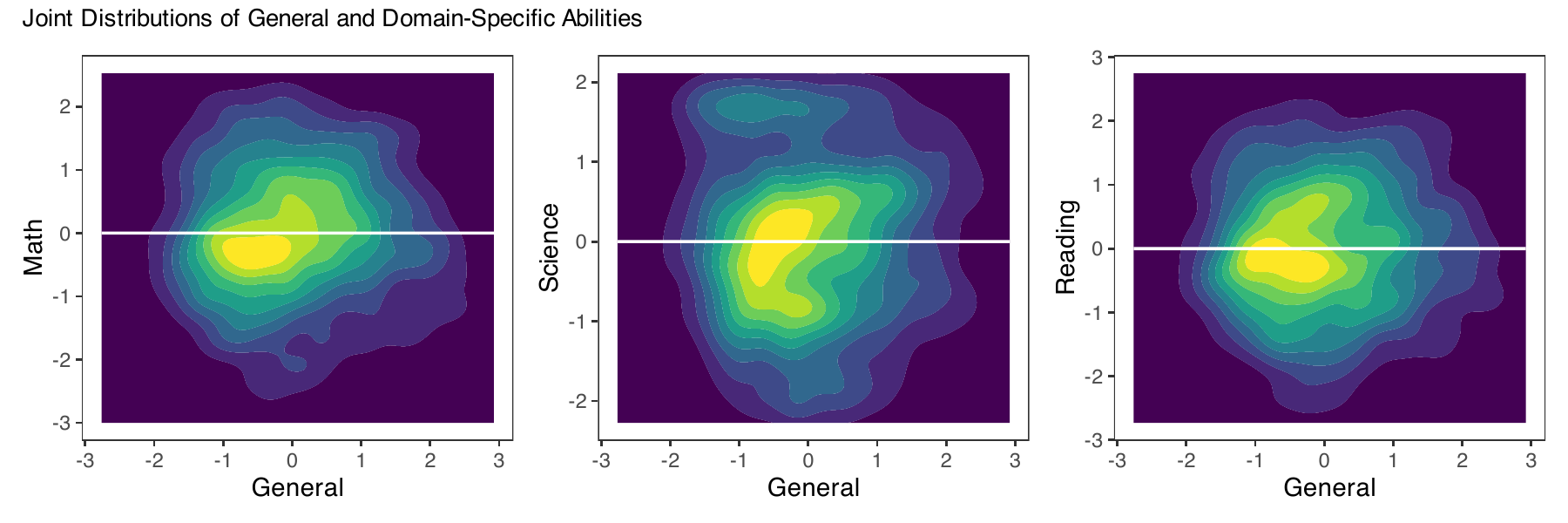}
    \includegraphics[width=5.6in, height=1.9in]{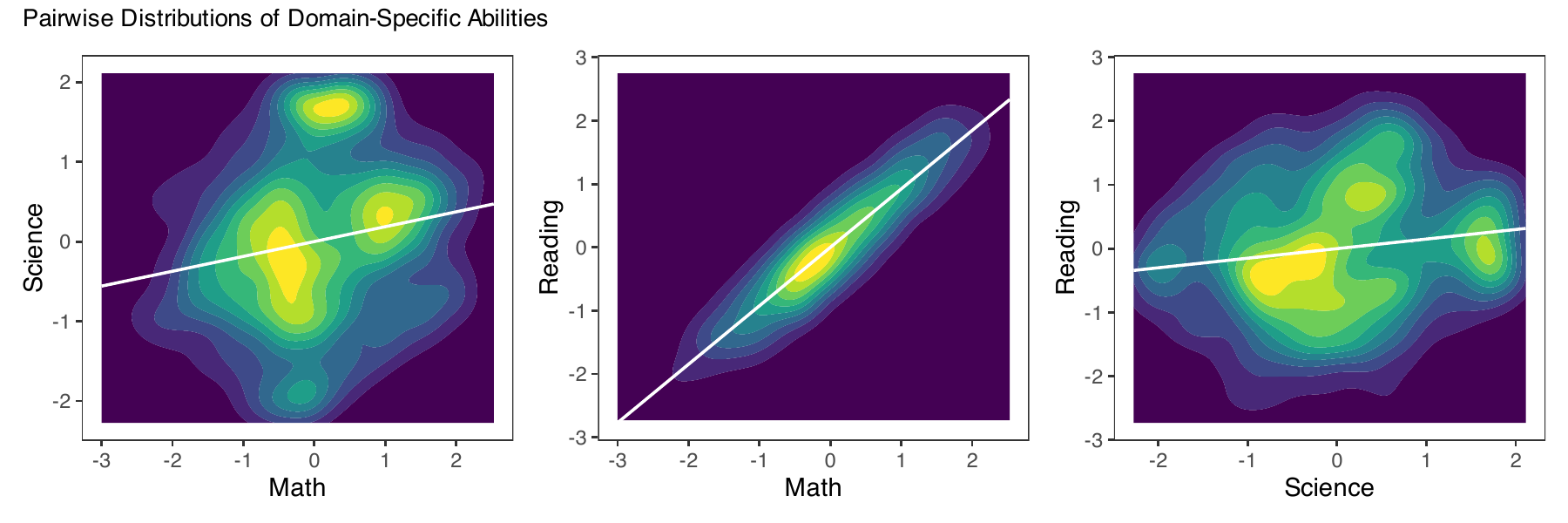}
    \caption{\small Joint distributions of the estimated latent abilities. The upper panel shows the joint distributions between the general ability and each domain-specific ability (mathematics, science, and reading), while the lower panel presents the pairwise relationships among the three domain-specific abilities. }
    \label{fig_pisa_factors}
\end{figure}}
Figure~\ref{fig_pisa_factors} presents the joint density contours of the estimated latent abilities. 
The general, mathematics, and reading factors exhibit approximately normal distributions, whereas the science factor shows two distinct density peaks, suggesting a non-unimodal distribution in this domain. Correlations between the science factor and the mathematics and reading factors are moderate. A notable feature is the pronounced linear ridge in the joint distribution of the mathematics and reading factors, indicating a strong correlation between these two abilities. This finding aligns with prior research \citep{pokropek2022much,organisation2023pisa}, pointing to potential dependence between mathematical reasoning and reading comprehension.


\section{Conclusion}\label{sec:end}
This paper provides a comprehensive statistical analysis of block structured latent variable models, covering model identifiability, estimation consistency, and statistical inference for the model parameters. We establish identifiability conditions for a broad class of block structures. We then consider maximum likelihood estimation subject to the proposed constraints. To analyze this constrained optimization problem, we propose a Lagrangian-type formulation and show that the minimizer of the Lagrangian function coincides with the solution of the original problem. Building on this result, we derive non‑asymptotic $\ell_2$- and $\ell_{\infty}$-error bounds and asymptotic distributions for the constrained maximum likelihood estimation. Theoretical findings are supported by numerical studies.

Our work opens several directions for future studies. 
First, in this paper, the link function $g_{ij}(\cdot\mid\cdot)$ is treated as known and correctly specified. An important direction for future research is to study robustness to misspecification and to develop semiparametric or nonparametric methods to estimate the link~\citep[e.g.][]{liu2022semiparametric,caner2025deep}. Second, our theoretical results assume a known number of latent factors. The analysis of the maximum likelihood estimation can be used to further develop procedures, such as information criteria, to determine the latent dimension and to select among competing models~\citep{chen2022determining}. 
Third, our framework provides a foundation for extending the analysis to more complex settings, such as multi-layer latent factor models~\citep{qiao2025exploratory} and deep generative models~\citep{moran2021identifiable}.

{
 \spacingset{1.33}
\setlength{\bibsep}{0pt} 
 \small
 \bibstyle{agsm}
\bibliography{reference}}

\end{document}